\documentclass[sigconf]{acmart}

\usepackage{booktabs} 
\usepackage{url}
\usepackage{amsmath}
\usepackage{amsthm}
\usepackage{amssymb}
\usepackage{graphicx}
\usepackage{setspace}
\usepackage{array}
\usepackage{float}
\usepackage{textcomp}
\usepackage{multirow}

\copyrightyear{2018} 
\acmYear{2018} 
\setcopyright{acmcopyright}
\acmConference[CCS '18]{2018 ACM SIGSAC Conference on Computer and Communications Security}{October 15--19, 2018}{Toronto, ON, Canada}
\acmBooktitle{2018 ACM SIGSAC Conference on Computer and Communications Security (CCS '18), October 15--19, 2018, Toronto, ON, Canada}
\acmPrice{15.00}
\acmDOI{10.1145/3243734.3243750}
\acmISBN{978-1-4503-5693-0/18/10}

\providecommand{\tabularnewline}{\\}
\floatstyle{ruled}
\newfloat{algorithm}{tbp}{loa}
\providecommand{\algorithmname}{Algorithm}
\floatname{algorithm}{\protect\algorithmname}

  \theoremstyle{acmdefinition}
  \newtheorem{defn}{\protect\definitionname}
\theoremstyle{acmplain}
\newtheorem{thm}{\protect\theoremname}
  \theoremstyle{acmplain}
  \newtheorem{lem}{\protect\lemmaname}
  \theoremstyle{remark}
  \newtheorem{rem}{\protect\remarkname}

  \providecommand{\definitionname}{Definition}
  \providecommand{\lemmaname}{Lemma}
  \providecommand{\remarkname}{Remark}
\providecommand{\theoremname}{Theorem}

\begin{document}
\title{MVG Mechanism: Differential Privacy under Matrix-Valued Query}

\author{Thee Chanyaswad}
\authornote{Currently at KBTG Machine Learning Team, Thailand, E-mail: theerachai.c@kbtg.tech}
\affiliation{%
    \institution{Princeton University, USA}
}
\email{tc7@princeton.edu}

\author{Alex Dytso}
\affiliation{%
  \institution{Princeton University, USA}
}
\email{adytso@princeton.edu}

\author{H. Vincent Poor}
\affiliation{%
  \institution{Princeton University, USA}
  }
\email{poor@princeton.edu}

\author{Prateek Mittal}
\affiliation{%
   \institution{Princeton University, USA}
}
\email{pmittal@princeton.edu}

\begin{abstract}
Differential privacy mechanism design has traditionally been tailored
for a scalar-valued query function. Although many mechanisms such
as the Laplace and Gaussian mechanisms can be extended to a matrix-valued
query function by adding i.i.d. noise to each element of the matrix,
this method is often suboptimal as it forfeits an opportunity to exploit
the structural characteristics typically associated with matrix analysis.
To address this challenge, we propose a novel differential privacy
mechanism called the \emph{Matrix-Variate Gaussian (MVG) mechanism},
which adds a \emph{matrix-valued} noise drawn from a matrix-variate
Gaussian distribution, and we rigorously prove that the MVG mechanism
preserves $(\epsilon,\delta)$-differential privacy. Furthermore,
we introduce the concept of \emph{directional noise} made possible
by the design of the MVG mechanism. Directional noise allows the impact
of the noise on the utility of the matrix-valued query function to
be moderated. Finally, we experimentally demonstrate the performance
of our mechanism using three matrix-valued queries on three privacy-sensitive
datasets. We find that the MVG mechanism can notably outperforms four
previous state-of-the-art approaches, and provides comparable utility
to the non-private baseline.
\end{abstract}

%
%



\maketitle

\section{Introduction}

Differential privacy \cite{RefWorks:186,RefWorks:195} has become
the gold standard for a rigorous privacy guarantee. This has prompted
the development of many mechanisms including the classical Laplace
mechanism \cite{RefWorks:195} and the Exponential mechanism \cite{RefWorks:192}.
In addition, there are other mechanisms that build upon these two
classical ones such as those based on data partition and aggregation
\cite{RefWorks:224,RefWorks:241,RefWorks:242,RefWorks:243,RefWorks:244,RefWorks:245,RefWorks:246,RefWorks:247,RefWorks:248,RefWorks:219,RefWorks:193},
and those based on adaptive queries \cite{RefWorks:236,RefWorks:221,RefWorks:316,RefWorks:317,RefWorks:318,RefWorks:319,RefWorks:175}.
From this observation, differentially-private mechanisms may be categorized
into the basic and derived mechanisms. Privacy guarantee of the basic
mechanisms is self-contained, whereas that of the derived mechanisms
is achieved through a combination of basic mechanisms, composition
theorems, and the post-processing invariance property \cite{RefWorks:185}.

In this work, we design a \emph{basic mechanism} for \emph{matrix-valued
queries}. Existing basic mechanisms for differential privacy are designed
typically for scalar-valued query functions. However, in many practical
settings, the query functions are multi-dimensional and can be succinctly
represented as matrix-valued functions. Examples of matrix-valued
query functions in the real-world applications include the covariance
matrix \cite{RefWorks:249,RefWorks:194,RefWorks:178}, the kernel
matrix \cite{RefWorks:33}, the adjacency matrix \cite{RefWorks:329},
the incidence matrix \cite{RefWorks:329}, the rotation matrix \cite{RefWorks:346},
the Hessian matrix \cite{RefWorks:347}, the transition matrix \cite{RefWorks:348},
and the density matrix \cite{RefWorks:349}, which find applications
in statistics \cite{RefWorks:350}, machine learning \cite{RefWorks:376},
graph theory \cite{RefWorks:329}, differential equations \cite{RefWorks:347},
computer graphics \cite{RefWorks:346}, probability theory \cite{RefWorks:348},
quantum mechanics \cite{RefWorks:349}, and many other fields \cite{RefWorks:328}.

One property that distinguishes the matrix-valued query functions
from the scalar-valued query functions is the relationship and interconnection
among the elements of the matrix. One may naively treat these matrices
as merely a collection of scalar values, but that could prove suboptimal
since the \emph{structure} and \emph{relationship} among these scalar
values are often informative and essential to the understanding and
analysis of the system. For example, in graph theory, the adjacency
matrix is symmetric for an undirected graph, but not for a directed
graph \cite{RefWorks:329} \textendash{} an observation which is implausible
to extract from simply looking at the collection of elements without
considering how they are arranged in the matrix.

In differential privacy, the standard method for a matrix-valued
query function is to extend a scalar-valued mechanism by adding \emph{independent
and identically distributed} (i.i.d.) noise to each element of the
matrix \cite{RefWorks:195,RefWorks:186,RefWorks:220}. However, this
method may not be optimal as it fails to utilize the structural characteristics
of the matrix-valued noise and query function. Although some advanced
methods have explored this possibility in an iterative/procedural
manner \cite{RefWorks:236,RefWorks:221,RefWorks:235}, the structural characteristics
of the matrices are still largely under-investigated. This is partly
due to the lack of a basic mechanism that is directly designed for
matrix-valued query functions, making such utilization and application
of available tools in matrix analysis challenging.

In this work, we formalize the study of the matrix-valued differential
privacy, and present a new basic mechanism that can readily exploit
the structural characteristics of the matrices \textendash{} the \emph{Matrix-Variate
Gaussian (MVG) mechanism}. The high-level concept of the MVG mechanism
is simple \textendash{} it adds a matrix-variate Gaussian noise scaled
to the $L_{2}$-sensitivity of the matrix-valued query function (cf.
Fig. \ref{fig:mvg_schematic}). We rigorously prove that the MVG mechanism
guarantees $(\epsilon,\delta)$-differential privacy. Moreover, due
to the multi-dimensional nature of the noise and the query function,
the MVG mechanism allows flexibility in the design via the novel notion
of \emph{directional noise}. An important consequence of the concept
of directional noise is that the matrix-valued noise in the MVG mechanism
can be devised to affect certain parts of the matrix-valued query
function less than the others, while providing \emph{the same} privacy
guarantee. In practice, this property could be beneficial as the noise
can be tailored to minimally impact the intended utility.

\begin{figure}
\begin{centering}
\includegraphics[scale=0.38]{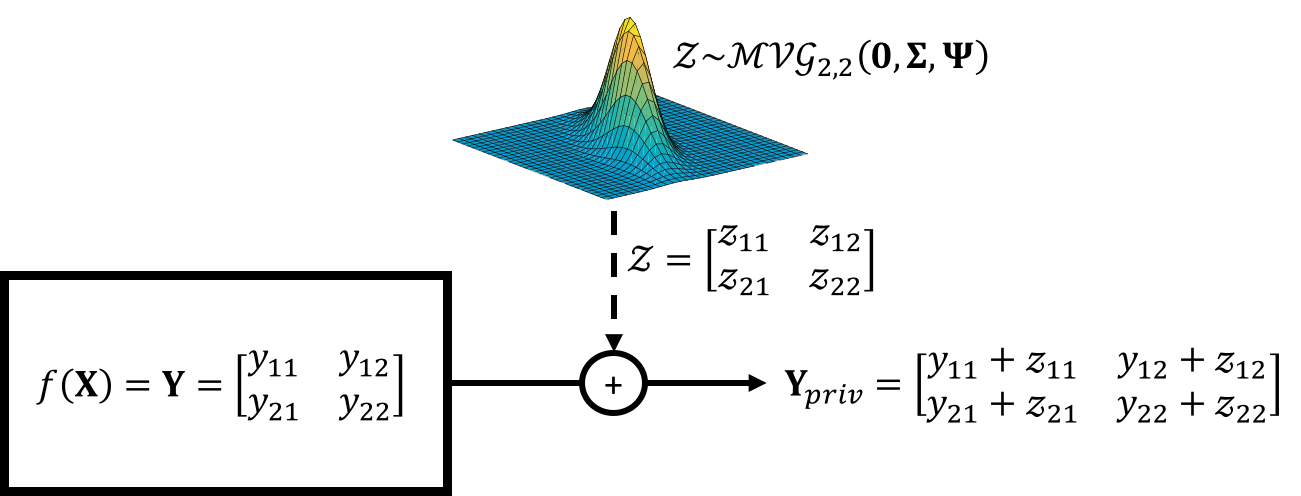} 
\par\end{centering}
\caption{Given a matrix-valued query function $f(\mathbf{X})\in\mathbb{R}^{m\times n}$,
the MVG mechanism adds a \emph{matrix-valued noise} drawn from the
\emph{matrix-variate Gaussian distribution} $\mathcal{MVG}_{m,n}(\mathbf{0},\boldsymbol{\Sigma},\boldsymbol{\Psi})$
to guarantee $(\epsilon,\delta)$-differential privacy. The schematic
shows an example when $m=n=2$. \label{fig:mvg_schematic}}

\end{figure}

Finally, to illustrate the effectiveness of the MVG mechanism, we
conduct experiments on three privacy-sensitive real-world datasets
\textendash{} Liver Disorders \cite{RefWorks:322,RefWorks:413}, Movement
Prediction \cite{RefWorks:396}, and Cardiotocography \cite{RefWorks:322,RefWorks:412}.
The experiments include three tasks involving matrix-valued query
functions \textendash{} regression, finding the first principal component,
and covariance estimation. The results show that the MVG mechanism
can outperform four prior state-of-the-art mechanisms \textendash{}
the Laplace mechanism, the Gaussian mechanism, the Exponential mechanism, and the
JL transform \textendash{} in utility in all experiments. 

To summarize, our main contributions are as follows. 
\begin{itemize}
\item  We formalize the study of matrix-valued query functions
in differential privacy and introduce the novel \emph{Matrix-Variate
Gaussian (MVG) mechanism}. 
\item We rigorously prove that the MVG mechanism guarantees $(\epsilon,\delta)$-differential
privacy. 
\item We introduce a novel concept of \emph{directional noise}, and propose
two simple algorithms to implement this novel concept with the MVG
mechanism. 
\item We evaluate our approach on three real-world datasets and show that
our approach can outperform four prior mechanisms in all experiments,
and yields utility close to the non-private baseline. 
\end{itemize}

\section{Prior Works} \label{sec:Prior-Works}

Existing mechanisms for differential privacy may be categorized into
two types: the \emph{basic mechanism} \cite{RefWorks:195,RefWorks:186,RefWorks:313,RefWorks:192,RefWorks:220,RefWorks:277,RefWorks:397,RefWorks:398,RefWorks:399}; and the \emph{derived mechanism} \cite{RefWorks:193,RefWorks:401,RefWorks:224,RefWorks:241,RefWorks:242,RefWorks:243,RefWorks:244,RefWorks:245,RefWorks:246,RefWorks:247,RefWorks:248,RefWorks:219,RefWorks:253,RefWorks:187,RefWorks:221,RefWorks:236,RefWorks:237,RefWorks:314,pca-gauss,RefWorks:400,RefWorks:337,RefWorks:339,RefWorks:316,RefWorks:317,RefWorks:236,RefWorks:175,RefWorks:319}. Since
our work concerns the basic mechanism design, we focus our discussion
on this type, and provide a general overview of the other.

\subsection{Basic Mechanisms} \label{subsec:Basic-Mechanisms}

Basic mechanisms are those whose privacy guarantee is self-contained,
i.e. it does not deduce the guarantee from another mechanism. Here,
we discuss four popular existing basic mechanisms.

\subsubsection{Laplace Mechanism.}

The classical \emph{Laplace mechanism} \cite{RefWorks:195} adds noise
drawn from the Laplace distribution scaled to the $L_{1}$-sensitivity
of the query function. It was initially designed for a scalar-valued
query function, but can be extended to a matrix-valued query function
by adding i.i.d. Laplace noise to each element of the matrix. The
Laplace mechanism provides the strong $\epsilon$-differential privacy
guarantee and is relatively simple to implement. However, its generalization
to a matrix-valued query function does not automatically utilize the
structure of the matrices involved.

\subsubsection{Gaussian Mechanism.}

The \emph{Gaussian mechanism} \cite{RefWorks:220,RefWorks:186,RefWorks:277}
uses i.i.d. additive noise drawn from the Gaussian distribution scaled
to the $L_{2}$-sensitivity. The Gaussian mechanism guarantees $(\epsilon,\delta)$-differential
privacy. Like the Laplace mechanism, it also does not automatically
consider the structure of the matrices.

\subsubsection{Johnson-Lindenstrauss (JL) Transform.}

The \emph{JL transform method} \cite{RefWorks:313} uses multiplicative
noise to guarantee $(\epsilon,\delta)$-differential privacy. It is,
in fact, a rare basic mechanism designed for a matrix-valued query
function. Despite its promise, previous works show that the JL transform
method can be applied to queries with certain properties only, e.g.
\begin{itemize}
\item Blocki et al. \cite{RefWorks:313} use a random matrix, whose
elements are drawn i.i.d. from a Gaussian distribution, and the method
is applicable to the Laplacian of a graph and the covariance matrix;
\item Blum and Roth \cite{RefWorks:397} use a hash function that implicitly
represents the JL transform, and the method is suitable for a sparse
query; 
\item Upadhyay \cite{RefWorks:398,RefWorks:399} uses a multiplicative
combination of random matrices to provide a JL transform that is applicable
to any matrix-valued query function whose singular values are all
above a threshold.
\end{itemize}

Among these methods, Upadhyay's works \cite{RefWorks:398,RefWorks:399}
stand out as possibly the most general. In our experiments, we show
that our approach can yield higher utility for the same privacy budget
than these methods.

\subsubsection{Exponential Mechanism.}

The \emph{Exponential mechanism} uses noise introduced via the sampling
process \cite{RefWorks:192}. It draws its query answers from a custom
distribution designed to preserve $\epsilon$-differential privacy.
To provide reasonable utility, the distribution is chosen based on
the \emph{quality function}, which indicates the utility score of
each possible sample. Due to its generality, it has been utilized
for many types of query functions, including the matrix-valued query
functions. We experimentally compare our approach to the Exponential
mechanism, and show that, with slightly weaker privacy guarantee,
our method can yield significant utility improvement.

\paragraph{Summary}
Finally, we conclude that our method differs from the four existing basic mechanisms as follows. In contrast with the i.i.d. noise in the Laplace and Gaussian mechanisms, the MVG mechanism allows a non-i.i.d. noise (cf. Sec. \ref{sec:Directional-Noise}). As opposed to the multiplicative noise in the JL transform and the sampling noise in the Exponential mechanism, the MVG mechanism uses an additive noise for matrix-valued query functions.

\subsection{Derived Mechanisms}

Derived mechanisms \textendash{} also referred to as ``revised algorithms'' by Blocki et al. \cite{RefWorks:313}
\textendash{} are those whose privacy guarantee is deduced from other
basic mechanisms via the composition theorems and the post-processing
invariance property \cite{RefWorks:185}. Derived mechanisms are often
designed to provide better utility by exploiting some properties of
the query function or the data.

The general techniques used by derived mechanisms are often translatable among basic mechanisms, including our MVG mechanism. Given our focus on a novel basic mechanism, these techniques are less relevant to our work, and we leave the investigation of integrating them into the MVG framework in the future work, and some of the popular techniques used by derived mechanisms are summarized here.

\subsubsection{Sensitivity Control.}

This technique avoids the worst-case sensitivity in basic mechanisms
by using variant concepts of sensitivity. Examples include the smooth
sensitivity \cite{RefWorks:193} and elastic sensitivity \cite{RefWorks:401}.

\subsubsection{Data Partition and Aggregation.}

This technique uses data partition and aggregation to produce more
accurate query answers \cite{RefWorks:224,RefWorks:241,RefWorks:242,RefWorks:243,RefWorks:244,RefWorks:245,RefWorks:246,RefWorks:247,RefWorks:248,RefWorks:219}.
The partition and aggregation processes are done in a differentially-private
manner either via the composition theorems and the post-processing
invariance property \cite{RefWorks:185}, or with a small extra privacy
cost. Hay et al. \cite{RefWorks:253} nicely summarize many works that utilize this concept.

\subsubsection{Non-uniform Data Weighting.}

This technique lowers the level of perturbation required for the privacy protection by \emph{weighting each data sample or dataset differently} \cite{RefWorks:187,RefWorks:221,RefWorks:236,RefWorks:237}.
The rationale is that each sample in a dataset, or each instance
of the dataset itself, has a heterogeneous contribution to the query output.
Therefore, these mechanisms place a higher weight on the
critical samples or instances of the database to provide better utility.

\subsubsection{Data Compression.}

This approach reduces the level of perturbation required for differential
privacy via \emph{dimensionality reduction}. Various dimensionality reduction methods have been proposed. For example, Kenthapadi et al. \cite{RefWorks:314}, Xu et al. \cite{RefWorks:400}, Li et al. \cite{RefWorks:337} and Chanyaswad et al. \cite{pca-gauss} use random projection; Jiang et al. \cite{RefWorks:339} use
principal component analysis (PCA) and linear discriminant analysis (LDA); Xiao et al. \cite{RefWorks:224}
use wavelet transform; and Acs et al. \cite{RefWorks:219} use lossy
Fourier transform.

\subsubsection{Adaptive Queries.}

These methods use past/auxiliary information to improve the utility
of the query answers. Examples are the matrix mechanism \cite{RefWorks:316,RefWorks:317},
the multiplicative weights mechanism \cite{RefWorks:236,RefWorks:221},
the low-rank mechanism \cite{RefWorks:318,RefWorks:539}, correlated noise \cite{RefWorks:235,RefWorks:538}, least-square estimation \cite{RefWorks:235}, boosting \cite{RefWorks:175},
and the sparse vector technique \cite{RefWorks:220,RefWorks:319}. 
We also note that some of these adaptive methods can be used  
in the restricted case of matrix-valued query where the 
matrix-valued query can be decomposed into multiple \emph{linear} 
vector-valued queries~\cite{RefWorks:235,RefWorks:537,RefWorks:539,yuan2016convex,nikolov2015improved,RefWorks:236}. However, such as an 
approach does not generalize for \emph{arbitrary} matrix-valued queries.

\paragraph{Summary}
We conclude with the following three main observations. (a) First, the MVG
mechanism falls into the category of basic mechanism. (b) Second, techniques
used in derived mechanisms are generally applicable to multiple basic
mechanisms, including our novel MVG mechanism. 
(c) Finally, therefore, for
fair comparison, we will compare the MVG mechanism with the four state-of-the-art
basic mechanisms presented in this section.

\section{Background}

\subsection{Matrix-Valued Query} \label{subsec:Matrix-Valued-Query}

We use the term dataset interchangeably with database, and represent
it with the data matrix $\mathbf{X}\in \mathbb{R}^{M\times N}$, whose columns are the $M$-dimensional vector samples/records. The matrix-valued query function,
$f(\mathbf{X})\in\mathbb{R}^{m\times n}$, has $m$ rows and $n$
columns \footnote{Note that we use the capital $M,N$ for the dimension of the dataset, but the small $m,n$ for the dimension of the query output.}. We define the notion of neighboring datasets $\{\mathbf{X}_{1},\mathbf{X}_{2}\}$
as two datasets that differ by a single record, and denote it as $d(\mathbf{X}_{1},\mathbf{X}_{2})=1$.
We note, however, that although the neighboring datasets differ by
only a single record, $f(\mathbf{X}_{1})$ and $f(\mathbf{X}_{2})$
may differ in every element.

We denote a matrix-valued random variable with the calligraphic font,
e.g. $\mathcal{Z}$, and its instance with the bold font, e.g. $\mathbf{Z}$.
Finally, as will become relevant later, we use the columns of $\mathbf{X}$
to denote the samples in the dataset.

\subsection{\texorpdfstring{$(\epsilon,\delta)$}{(epsilon,delta)}-Differential Privacy}

Differential privacy \cite{RefWorks:151,RefWorks:186} guarantees
that the involvement of any one particular record of the dataset would
not drastically change the query answer. 
\begin{defn}
\label{def:differential_privacy}A mechanism $\mathcal{A}$ on a query
function $f(\cdot)$ is $(\epsilon,\delta)$- differentially-private
if for all neighboring datasets $\{\mathbf{X}_{1},\mathbf{X}_{2}\}$,
and for all possible measurable matrix-valued outputs $\mathbf{S}\subseteq\mathbb{R}^{m\times n}$,
\[
\Pr[\mathcal{A}(f(\mathbf{X}_{1}))\in\mathbf{S}]\leq e^{\epsilon}\Pr[\mathcal{A}(f(\mathbf{X}_{2}))\in\mathbf{S}]+\delta.
\]
\end{defn}

\subsection{Matrix-Variate Gaussian Distribution}

One of our main innovations is the use of the noise drawn from a matrix-variate
probability distribution. Specifically, in the MVG mechanism,
the additive noise is drawn from the matrix-variate Gaussian distribution,
defined as follows \cite{RefWorks:279,RefWorks:280,RefWorks:281,RefWorks:282,RefWorks:284,RefWorks:367}. 
\begin{defn}
\label{def:tmvg_dist}An $m\times n$ matrix-valued random variable
$\mathcal{X}$ has a matrix-variate Gaussian distribution $\mathcal{MVG}_{m,n}(\mathbf{M},\boldsymbol{\Sigma},\boldsymbol{\Psi})$,
if it has the density function: 
\[
p_{\mathcal{X}}(\mathbf{X})=\frac{\exp\{-\frac{1}{2}\mathrm{tr}[\boldsymbol{\Psi}^{-1}(\mathbf{X}-\mathbf{M})^{T}\boldsymbol{\Sigma}^{-1}(\mathbf{X}-\mathbf{M})]\}}{(2\pi)^{mn/2}\left|\boldsymbol{\Psi}\right|^{m/2}\left|\boldsymbol{\Sigma}\right|^{n/2}},
\]
where $\mathrm{tr}(\cdot)$ is the matrix trace \cite{RefWorks:208},
$\left|\cdot\right|$ is the matrix determinant \cite{RefWorks:208},
$\mathbf{M}\in\mathbb{R}^{m\times n}$ is the mean, $\boldsymbol{\Sigma}\in\mathbb{R}^{m\times m}$
is the row-wise covariance, and $\boldsymbol{\Psi}\in\mathbb{R}^{n\times n}$
is the column-wise covariance. 
\end{defn}

Notably, the density function of $\mathcal{MVG}_{m,n}(\mathbf{M},\boldsymbol{\Sigma},\boldsymbol{\Psi})$
looks similar to that of the multivariate Gaussian, $\mathcal{N}_{m}(\boldsymbol{\mu},\boldsymbol{\Sigma})$.
Indeed, the matrix-variate Gaussian distribution $\mathcal{MVG}_{m,n}(\mathbf{M},\boldsymbol{\Sigma},\boldsymbol{\Psi})$
is a generalization of $\mathcal{N}_{m}(\boldsymbol{\mu},\boldsymbol{\Sigma})$
to a matrix-valued random variable. This leads to a few notable additions.
First, the mean vector $\boldsymbol{\mu}$ now becomes the mean matrix
$\mathbf{M}$. Second, in addition to the traditional row-wise covariance
matrix $\boldsymbol{\Sigma}$, there is also the column-wise covariance
matrix $\boldsymbol{\Psi}$. The latter is due to the fact that, not
only could the rows of the matrix be distributed non-uniformly, but
also could its columns.

We may intuitively explain this addition as follows. If we draw $n$
i.i.d. samples from $\mathcal{N}_{m}(\mathbf{0},\boldsymbol{\Sigma})$
denoted as $\mathbf{y}_{1},\ldots,\mathbf{y}_{n}\in\mathbb{R}^{m}$,
and concatenate them into a matrix $\mathbf{Y}=[\mathbf{y}_{1},\ldots,\mathbf{y}_{n}]\in\mathbb{R}^{m\times n}$,
then, it can be shown that $\mathbf{Y}$ is drawn from $\mathcal{MVG}_{m,n}(\mathbf{0},\boldsymbol{\Sigma},\mathbf{I})$,
where $\mathbf{I}$ is the identity matrix \cite{RefWorks:279}. However,
if we consider the case when the columns of $\mathbf{Y}$ are not
i.i.d., and are distributed with the covariance $\boldsymbol{\Psi}$
instead, then, it can be shown that this is distributed according
to $\mathcal{MVG}_{m,n}(\mathbf{0},\boldsymbol{\Sigma},\boldsymbol{\Psi})$
\cite{RefWorks:279}.

\subsection{Relevant Matrix Algebra Theorems}

We recite two major theorems in linear algebra that are essential
to the subsequent analysis. The first one is used in multiple parts of the analysis including the privacy proof and the interpretation of the results, while the second one is the concentration bound essential to the privacy proof. 
\begin{thm}[Singular value decomposition (SVD) \cite{RefWorks:208}]
\label{thm:svd}A matrix $\mathbf{A}\in\mathbb{R}^{m\times n}$ can
be decomposed as $\mathbf{A}=\mathbf{W}_{1}\boldsymbol{\Lambda}\mathbf{W}_{2}^{T}$,
where $\mathbf{W}_{1}\in\mathbb{R}^{m\times m},\mathbf{W}_{2}\in\mathbb{R}^{n\times n}$
are unitary, and $\boldsymbol{\Lambda}$ is a diagonal matrix whose
diagonal elements are the ordered \emph{singular values} of $\mathbf{A}$,
denoted as $\sigma_{1}\geq\sigma_{2}\geq\cdots\geq0$. 
\end{thm}
\begin{thm}[Laurent-Massart \cite{RefWorks:369}]
\label{thm:laurent_massart} For a matrix-variate random variable
$\mathcal{N}\sim\mathcal{MVG}_{m,n}(\mathbf{0},\mathbf{I}_{m},\mathbf{I}_{n})$,
$\delta\in[0,1]$, and $\zeta(\delta)=2\sqrt{-mn\ln\delta}-2\ln\delta+mn$,
\[
\Pr[\left\Vert \mathcal{N}\right\Vert _{F}^{2}\leq\zeta(\delta)^{2}]\geq1-\delta.
\]
\end{thm}

\section{MVG Mechanism: Differential Privacy with Matrix-Valued Query}

Matrix-valued query functions are different from their scalar counterparts
in terms of the vital information contained in how the elements are
arranged in the matrix. To fully exploit these structural characteristics
of matrix-valued query functions, we present our novel mechanism for
matrix-valued query functions: the \emph{Matrix-Variate Gaussian (MVG)
mechanism}.

First, let us introduce the sensitivity of the matrix-valued query
function used in the MVG mechanism.
\begin{defn}[Sensitivity]
\label{def:sensitivity}Given a matrix-valued query function $f(\mathbf{X})\in\mathbb{R}^{m\times n}$,
define the $L_{2}$-sensitivity as, 
\[
s_{2}(f)=\sup_{d(\mathbf{X}_{1},\mathbf{X}_{2})=1}\left\Vert f(\mathbf{X}_{1})-f(\mathbf{X}_{2})\right\Vert _{F},
\]
where $\left\Vert \cdot\right\Vert _{F}$ is the Frobenius norm \cite{RefWorks:208}. 
\end{defn}
Then, we present the MVG mechanism as follows. 
\begin{defn}[MVG mechanism]
\label{def:mvg_mech}Given a matrix-valued query function $f(\mathbf{X})\in\mathbb{R}^{m\times n}$,
and a matrix-valued random variable $\mathcal{Z}\sim\mathcal{MVG}_{m,n}(\mathbf{0},\boldsymbol{\Sigma},\boldsymbol{\Psi})$,
the \emph{MVG mechanism} is defined as, 
\[
\mathcal{MVG}(f(\mathbf{X}))=f(\mathbf{X})+\mathcal{Z},
\]
where $\boldsymbol{\Sigma}$ is the row-wise covariance matrix, and
$\boldsymbol{\Psi}$ is the column-wise covariance matrix. 
\end{defn}
So far, we have not specified how to pick $\boldsymbol{\Sigma}$ and
$\boldsymbol{\Psi}$ according to the sensitivity $s_{2}(f)$ in the
MVG mechanism. We discuss the explicit form of $\boldsymbol{\Sigma}$
and $\boldsymbol{\Psi}$ next.

\begin{table}
\begin{centering}
\begin{tabular}{|>{\centering}p{2cm}|>{\centering}p{6cm}|}
\hline 
{$\mathbf{X}\in \mathbb{R}^{M\times N}$} & {database/dataset whose $N$ columns are data samples and $M$ 
rows are attributes/features.}\tabularnewline
\hline 
{\footnotesize{}$\mathcal{MVG}_{m,n}(\mathbf{0},\boldsymbol{\Sigma},\boldsymbol{\Psi})$} & { $m\times n$ matrix-variate Gaussian distribution
with zero mean, the row-wise covariance $\boldsymbol{\Sigma}$, and the column-wise
covariance $\boldsymbol{\Psi}$.}\tabularnewline
\hline 
{$f(\mathbf{X})\in\mathbb{R}^{m\times n}$}  & {matrix-valued query function}\tabularnewline
\hline 
 {$r$}  & {$\min\{m,n\}$}\tabularnewline
\hline 
{$H_{r}$}  & {generalized harmonic numbers of order $r$}\tabularnewline
\hline 
{$H_{r,1/2}$}  & {generalized harmonic numbers of order $r$ of $1/2$}\tabularnewline
\hline 
{$\gamma$}  & {$\sup_{\mathbf{X}}\left\Vert f(\mathbf{X})\right\Vert _{F}$}\tabularnewline
\hline 
{$\zeta(\delta)$}  & {$2\sqrt{-mn\ln\delta}-2\ln\delta+mn$}\tabularnewline
\hline 
{$\boldsymbol{\sigma}(\boldsymbol{\Sigma}^{-1})$}  & {vector of non-increasing singular values of $\boldsymbol{\Sigma}^{-1}$ }\tabularnewline
\hline 
{$\boldsymbol{\sigma}(\boldsymbol{\Psi}^{-1})$}  & {vector of non-increasing singular values of $\boldsymbol{\Psi}^{-1}$ }\tabularnewline
\hline 
\end{tabular}
\par\end{centering}
\caption{Notations for the differential privacy analysis.} \label{tab:Notations}
 
\end{table}

As the additive matrix-valued noise of the MVG mechanism is drawn
from $\mathcal{MVG}_{m,n}(\mathbf{0},\boldsymbol{\Sigma},\boldsymbol{\Psi})$,
the parameters to be designed for the MVG mechanism are the two covariance
matrices $\boldsymbol{\Sigma}$ and $\boldsymbol{\Psi}$. The following
theorem presents a sufficient condition for the values of $\boldsymbol{\Sigma}$
and $\boldsymbol{\Psi}$ to ensure that the MVG mechanism preserves
$(\epsilon,\delta)$-differential privacy. 

\begin{thm}
\label{thm:design_general} Let 

\[
\boldsymbol{\sigma}(\boldsymbol{\Sigma}^{-1})=[\sigma_{1}(\boldsymbol{\Sigma}^{-1}),\ldots,\sigma_{m}(\boldsymbol{\Sigma}^{-1})]^{T},
\]
and 
\[
\boldsymbol{\sigma}(\boldsymbol{\Psi}^{-1})=[\sigma_{1}(\boldsymbol{\Psi}^{-1}),\ldots,\sigma_{n}(\mathbf{\boldsymbol{\Psi}}^{-1})]^{T}
\]
be the vectors of non-increasingly ordered singular values of $\boldsymbol{\Sigma}^{-1}$
and $\boldsymbol{\Psi}^{-1}$, respectively, and let the relevant
variables be defined according to Table \ref{tab:Notations}. Then,
the MVG mechanism guarantees $(\epsilon,\delta)$-differential privacy
if $\boldsymbol{\Sigma}$ and $\boldsymbol{\Psi}$ satisfy the following
condition,\footnote{Note that the dependence on $\delta$ is via $\zeta(\delta)$ in $\beta$.}
\begin{equation}
\left\Vert \boldsymbol{\sigma}(\boldsymbol{\Sigma}^{-1})\right\Vert _{2}\left\Vert \boldsymbol{\sigma}(\boldsymbol{\Psi}^{-1})\right\Vert _{2}\leq\frac{(-\beta+\sqrt{\beta^{2}+8\alpha\epsilon})^{2}}{4\alpha^{2}},\label{eq.sufficient_condition}
\end{equation}
where $\alpha=[H_{r}+H_{r,1/2}]\gamma^{2}+2H_{r}\gamma s_{2}(f)$,
and $\beta=2(mn)^{1/4}H_{r}s_{2}(f)\zeta(\delta)$. 
\end{thm}
\begin{proof}
(Sketch) We only provide the sketch proof here. The full proof can
be found in Appendix \ref{sec:Full-Proof}. 

The MVG mechanism guarantees $(\epsilon,\delta)$-differential privacy
if for every pair of neighboring datasets $\{\mathbf{X}_{1},\mathbf{X}_{2}\}$
and all measurable sets $\mathbf{S}\subseteq\mathbb{R}^{m\times n}$,
\[
\Pr\left[f(\mathbf{X}_{1})+\mathcal{Z}\in\mathbf{S}\right]\leq\exp(\epsilon)\Pr\left[f(\mathbf{X}_{2})+\mathcal{Z}\in\mathbf{S}\right]+\delta.
\]
Using Definition \ref{def:tmvg_dist}, this is satisfied if we have,
\begin{align*}
\int_{\mathbf{S}}e^{-\frac{1}{2}\mathrm{tr}[\boldsymbol{\Psi}^{-1}(\mathbf{Y}-f(\mathbf{X}_{1}))^{T}\boldsymbol{\Sigma}^{-1}(\mathbf{Y}-f(\mathbf{X}_{1}))]}d\mathbf{Y}\leq\\
e^{\epsilon}\int_{\mathbf{S}}e^{-\frac{1}{2}\mathrm{tr}[\boldsymbol{\Psi}^{-1}(\mathbf{Y}-f(\mathbf{X}_{2}))^{T}\boldsymbol{\Sigma}^{-1}(\mathbf{Y}-f(\mathbf{X}_{2}))]}d\mathbf{Y}+\delta.
\end{align*}
By inserting $\frac{\exp\{-\frac{1}{2}\mathrm{tr}[\boldsymbol{\Psi}^{-1}(\mathbf{Y}-f(\mathbf{X}_{2}))^{T}\boldsymbol{\Sigma}^{-1}(\mathbf{Y}-f(\mathbf{X}_{2}))]\}}{\exp\{-\frac{1}{2}\mathrm{tr}[\boldsymbol{\Psi}^{-1}(\mathbf{Y}-f(\mathbf{X}_{2}))^{T}\boldsymbol{\Sigma}^{-1}(\mathbf{Y}-f(\mathbf{X}_{2}))]\}}$
inside the integral on the left side, it is sufficient to show that
\[
\frac{\exp\{-\frac{1}{2}\mathrm{tr}[\boldsymbol{\Psi}^{-1}(\mathcal{Y}-f(\mathbf{X}_{1}))^{T}\boldsymbol{\Sigma}^{-1}(\mathcal{Y}-f(\mathbf{X}_{1}))]\}}{\exp\{-\frac{1}{2}\mathrm{tr}[\boldsymbol{\Psi}^{-1}(\mathcal{Y}-f(\mathbf{X}_{2}))^{T}\boldsymbol{\Sigma}^{-1}(\mathcal{Y}-f(\mathbf{X}_{2}))]\}}\leq\exp(\epsilon),
\]
with probability $\geq1-\delta$. By algebraic manipulations, we can
express this condition as, 
\begin{align*}
\mathrm{tr}[\boldsymbol{\Psi}^{-1}\mathcal{Y}^{T}\boldsymbol{\Sigma}^{-1}\boldsymbol{\Delta}+\boldsymbol{\Psi}^{-1}\boldsymbol{\Delta}^{T}\boldsymbol{\Sigma}^{-1}\mathcal{Y}\\
+\boldsymbol{\Psi}^{-1}f(\mathbf{X}_{2})^{T}\boldsymbol{\Sigma}^{-1}f(\mathbf{X}_{2})-\boldsymbol{\Psi}^{-1}f(\mathbf{X}_{1})^{T}\boldsymbol{\Sigma}^{-1}f(\mathbf{X}_{1})] & \leq2\epsilon.
\end{align*}
where $\boldsymbol{\Delta}=f(\mathbf{X}_{1})-f(\mathbf{X}_{2})$.
This is the necessary condition that has to be satisfied for all neighboring
$\{\mathbf{X}_{1},\mathbf{X}_{2}\}$ with probability $\geq1-\delta$
for the MVG mechanism to guarantee $(\epsilon,\delta)$-differential
privacy. Therefore, we refer to it as the \emph{characteristic equation}.
From here, the proof analyzes the four terms in the sum separately
since the trace is additive. The analysis relies on the following
lemmas in linear algebra. 
\begin{lem}[Merikoski-Sarria-Tarazaga \cite{RefWorks:288}]
\label{lem:singular_bound} The non-increasingly ordered singular
values of a matrix $\mathbf{A}\in\mathbb{R}^{m\times n}$ have the
values of $0\leq\sigma_{i}\leq\left\Vert \mathbf{A}\right\Vert _{F}/\sqrt{i}$.
\end{lem}
\begin{lem}[von Neumann \cite{RefWorks:402}]
\label{lem:v_neumann}Let $\mathbf{A},\mathbf{B}\in\mathbb{R}^{m\times n}$;
$\sigma_{i}(\mathbf{A})$ and $\sigma_{i}(\mathbf{B})$ be the non-increasingly
ordered singular values of $\mathbf{A}$ and $\mathbf{B}$, respectively;
and $r=\min\{m,n\}$. Then, $\mathrm{tr}(\mathbf{A}\mathbf{B}^{T})\leq\Sigma_{i=1}^{r}\sigma_{i}(\mathbf{A})\sigma_{i}(\mathbf{B})$. 
\end{lem}
\begin{lem}[Trace magnitude bound \cite{RefWorks:278}]
\label{lem:abs_trace_bound}Let $\sigma_{i}(\mathbf{A})$ be the
non-increasingly ordered singular values of $\mathbf{A}\in\mathbb{R}^{m\times n}$,
and $r=\min\{m,n\}$. Then, $\left|\mathrm{tr}(\mathbf{A})\right|\leq\sum_{i=1}^{r}\sigma_{i}(\mathbf{A})$. 
\end{lem}
The proof, then, proceeds with the analysis of each term in the characteristic
equation as follows.

\emph{The first term}: $\mathrm{tr}[\boldsymbol{\Psi}^{-1}\mathcal{Y}^{T}\boldsymbol{\Sigma}^{-1}\boldsymbol{\Delta}]$.
Let us denote $\mathcal{Y}=f(\mathbf{X})+\mathcal{Z}$, where $f(\mathbf{X})$
and $\mathcal{Z}$ are any possible instances of the query and the
noise, respectively. Then, we can rewrite the first term as, $\mathrm{tr}[\boldsymbol{\Psi}^{-1}f(\mathbf{X})^{T}\boldsymbol{\Sigma}^{-1}\boldsymbol{\Delta}]+\mathrm{tr}[\boldsymbol{\Psi}^{-1}\mathcal{Z}^{T}\boldsymbol{\Sigma}^{-1}\boldsymbol{\Delta}]$.
Both parts are then bounded by their singular values via Lemma \ref{lem:v_neumann}.
The singular values are, in turn, bounded via Lemma \ref{lem:singular_bound}
and Theorem \ref{thm:laurent_massart} with probability $\geq1-\delta$.
This gives the bound for the first term: 
\[
\mathrm{tr}[\boldsymbol{\Psi}^{-1}\mathcal{Y}^{T}\boldsymbol{\Sigma}^{-1}\boldsymbol{\Delta}]\leq\gamma H_{r}s_{2}(f)\phi^{2}+(mn)^{1/4}\zeta(\delta)H_{r}s_{2}(f)\phi,
\]
where $\phi=(\left\Vert \boldsymbol{\sigma}(\boldsymbol{\Sigma}^{-1})\right\Vert _{2}\left\Vert \boldsymbol{\sigma}(\boldsymbol{\Psi}^{-1})\right\Vert _{2})^{1/2}$.

\emph{The second term}: $\mathrm{tr}[\boldsymbol{\Psi}^{-1}\boldsymbol{\Delta}^{T}\boldsymbol{\Sigma}^{-1}\mathcal{Y}]$.
By following the same steps as in the first term, the second term
has the exact same bound as the first term, i.e. 
\[
\mathrm{tr}[\boldsymbol{\Psi}^{-1}\boldsymbol{\Delta}^{T}\boldsymbol{\Sigma}^{-1}\mathcal{Y}]\leq\gamma H_{r}s_{2}(f)\phi^{2}+(mn)^{1/4}\zeta(\delta)H_{r}s_{2}(f)\phi.
\]

\emph{The third term}: $\mathrm{tr}[\boldsymbol{\Psi}^{-1}f(\mathbf{X}_{2})^{T}\boldsymbol{\Sigma}^{-1}f(\mathbf{X}_{2})]$.
Applying Lemma \ref{lem:v_neumann} and Lemma \ref{lem:singular_bound},
we can readily bound it as, 
\[
\mathrm{tr}[\boldsymbol{\Psi}^{-1}f(\mathbf{X}_{2})^{T}\boldsymbol{\Sigma}^{-1}f(\mathbf{X}_{2})]\leq\gamma^{2}H_{r}\phi^{2}.
\]

\emph{The fourth term}: $-\mathrm{tr}[\boldsymbol{\Psi}^{-1}f(\mathbf{X}_{1})^{T}\boldsymbol{\Sigma}^{-1}f(\mathbf{X}_{1})]$.
Since this term has the negative sign, we use Lemma \ref{lem:abs_trace_bound}
to bound its magnitude by its singular values. Then, we use Lemma
\ref{lem:singular_bound} to bound the singular values. This gives
the bound for the forth term as, 
\[
\left|\mathrm{tr}[\boldsymbol{\Psi}^{-1}f(\mathbf{X}_{1})^{T}\boldsymbol{\Sigma}^{-1}f(\mathbf{X}_{1})]\right|\leq\gamma^{2}H_{r,1/2}\phi^{2}.
\]

\emph{Four terms combined:} by combining the four terms and rearranging
them, the characteristic equation becomes $\alpha\phi^{2}+\beta\phi\leq2\epsilon$.
This is a quadratic equation, of which the solution is $\phi\in[\frac{-\beta-\sqrt{\beta^{2}+8\alpha\epsilon}}{2\alpha},\frac{-\beta+\sqrt{\beta^{2}+8\alpha\epsilon}}{2\alpha}]$.
Since we know $\phi>0$, we have the solution, 
\[
\phi\leq\frac{-\beta+\sqrt{\beta^{2}+8\alpha\epsilon}}{2\alpha},
\]
which implies the criterion in Theorem \ref{thm:design_general}. 
\end{proof}

\begin{figure}
\begin{centering}
\includegraphics[scale=0.4]{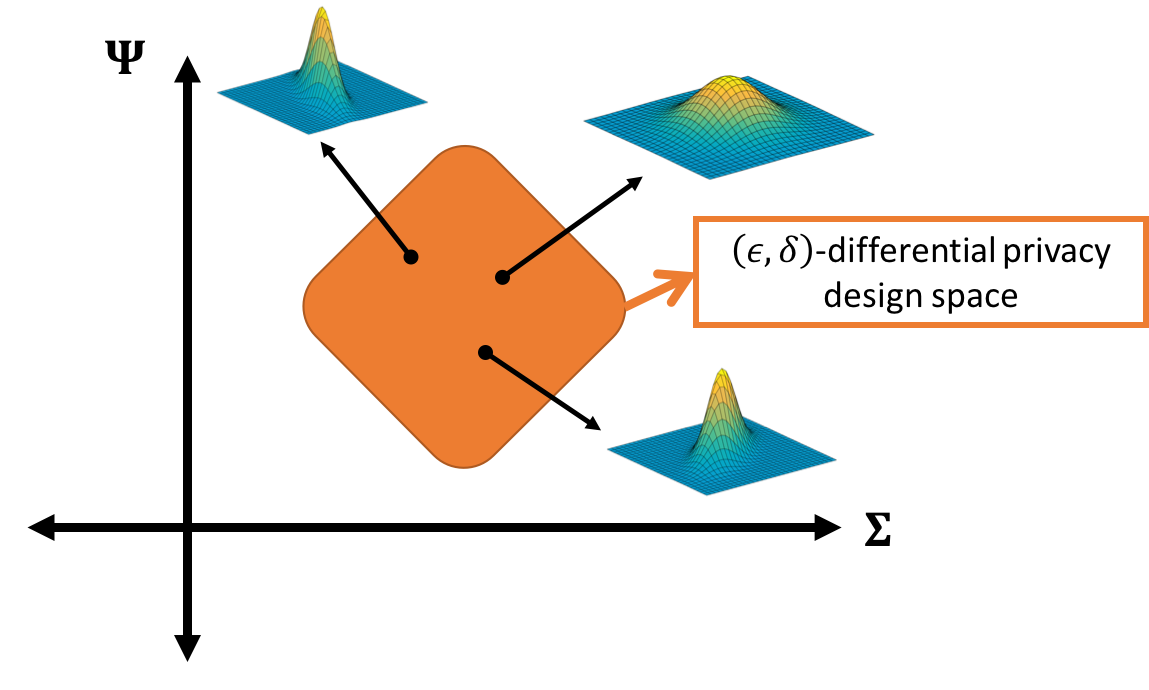} 
\par\end{centering}
\caption{A conceptual display of the MVG design space.
The illustration visualizes the design space coordinated by the two design parameters of $\mathcal{MVG}_{m,n}(\mathbf{0},\boldsymbol{\Sigma},\boldsymbol{\Psi})$. Each point on the space corresponds to an instance
of $\mathcal{MVG}_{m,n}(\mathbf{0},\boldsymbol{\Sigma},\boldsymbol{\Psi})$.
From this perspective, Theorem \ref{thm:design_general} suggests
that any instance of $\mathcal{MVG}_{m,n}(\mathbf{0},\boldsymbol{\Sigma},\boldsymbol{\Psi})$
in the (conceptual) shaded area would preserve $(\epsilon,\delta)$-differential
privacy. \label{fig:A-conceptual-display}}

\end{figure}

\begin{rem}
In Theorem \ref{thm:design_general}, we assume that the Frobenius
norm of the query function is bounded for all possible datasets by
$\gamma$. This assumption is valid in practice because real-world
data are rarely unbounded (cf. \cite{RefWorks:372}), and it is a
common assumption in the analysis of differential privacy for multi-dimensional
query functions (cf. \cite{RefWorks:195,RefWorks:178,RefWorks:338,RefWorks:249}). 
\end{rem}
\begin{rem}
The values of the generalized harmonic numbers \textendash{} $H_{r}$,
and $H_{r,1/2}$ \textendash{} can be obtained from the table lookup
for a given value of $r$, or can easily be computed recursively \cite{RefWorks:289}.
\end{rem}

The sufficient condition in Theorem \ref{thm:design_general} yields
an important observation: the privacy guarantee of the MVG mechanism
depends \emph{only on the singular values} of $\boldsymbol{\Sigma}$
and $\boldsymbol{\Psi}$. In other words, we may have multiple instances
of $\mathcal{MVG}_{m,n}(\mathbf{0},\boldsymbol{\Sigma},\boldsymbol{\Psi})$
that yield the exact same privacy guarantee (cf. Fig. \ref{fig:A-conceptual-display}).
To better understand this phenomenon, in the next section, we introduce the novel concept of \emph{directional noise}.

\section{Directional Noise} \label{sec:Directional-Noise}

Recall from Theorem \ref{thm:design_general} that the $(\epsilon,\delta)$-differential-privacy
condition for the MVG mechanism only applies to the singular values
of the two covariance matrices $\boldsymbol{\Sigma}$ and $\boldsymbol{\Psi}$.
Here, we investigate the ramification of this result via the
novel notion of \emph{directional noise}.

\subsection{Motivation for Non-i.i.d. Noise}

For a matrix-valued query function, the standard method for basic
mechanisms that use additive noise is by adding the \emph{independent
and identically distributed} (i.i.d.) noise to each element of the
matrix query. However, as common in matrix analysis \cite{RefWorks:208},
the matrices involved often have some geometric and algebraic characteristics
that can be exploited. As a result, it is usually the case that only
certain ``parts'' \textendash{} the term which will be defined more
precisely shortly \textendash{} of the matrices contain useful information.
In fact, this is one of the rationales behind many compression techniques
such as the popular principal component analysis (PCA) \cite{RefWorks:33,RefWorks:51,RefWorks:225}.
Due to this reason, adding the same amount of noise to every ``part''
of the matrix query may be highly suboptimal.

\subsection{Directional Noise as a Non-i.i.d. Noise}

\label{subsec:dir_noise_as_noniid}

\begin{figure}
\includegraphics[scale=0.28]{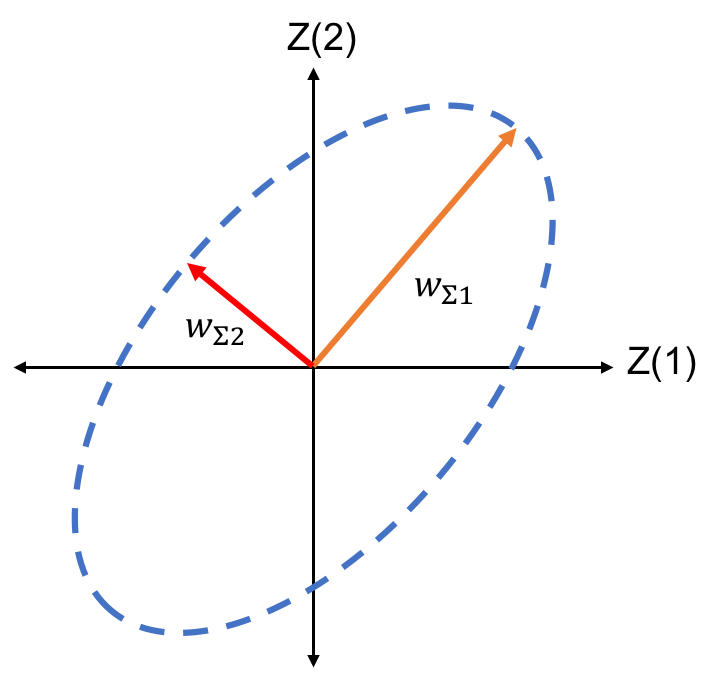}\includegraphics[scale=0.27]{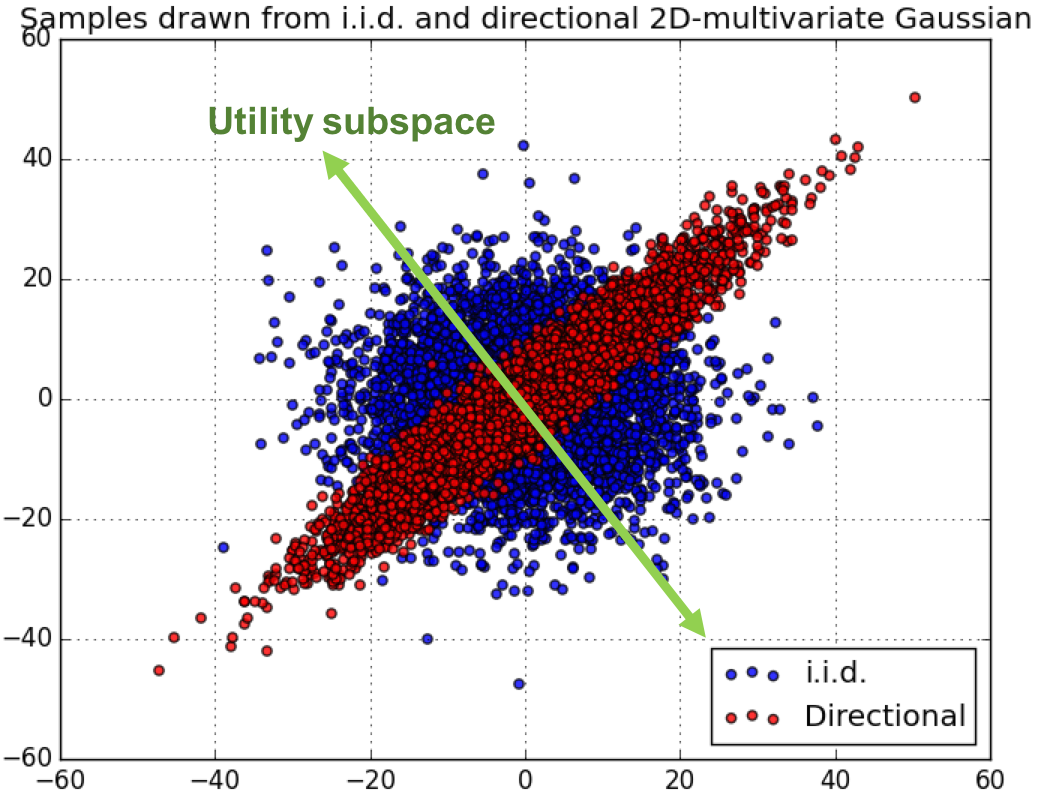}

\caption{(Left) An ellipsoid of equi-density contour of a 2D multivariate Gaussian
distribution. The two arrows indicate the principal axes of this ellipsoid.
(Right) Directional noise (red) and i.i.d. noise (blue) drawn from
a 2D-multivariate Gaussian distribution. The green line represents
a possible utility subspace that can benefit from this instance of
directional noise. \label{fig:dir_noise_samples}}

\end{figure}

Let us elaborate further on the ``parts'' of a matrix. In matrix
analysis, \emph{matrix factorization} \cite{RefWorks:208} is often
used to extract underlying properties of a matrix. This is a family
of algorithms and the specific choice depends upon the application
and types of insights it requires. Particularly, in our application,
the factorization that is informative is the singular value decomposition
(SVD) (Theorem \ref{thm:svd}) of the two covariance matrices of \emph{$\mathcal{MVG}_{m,n}(\mathbf{0},\boldsymbol{\Sigma},\boldsymbol{\Psi})$}.

Consider first the covariance matrix $\boldsymbol{\Sigma}\in\mathbb{R}^{m\times m}$,
and write its SVD as, $\boldsymbol{\Sigma}=\mathbf{W}_{1}\boldsymbol{\Lambda}\mathbf{W}_{2}^{T}$.
It is well-known that, for the covariance matrix, we have the equality
$\mathbf{W}_{1}=\mathbf{W}_{2}$ since it is positive definite (cf.
\cite{RefWorks:291,RefWorks:208}). Hence, let us more concisely write
the SVD of $\boldsymbol{\Sigma}$ as, 
\[
\boldsymbol{\Sigma}=\mathbf{W}_{\boldsymbol{\Sigma}}\boldsymbol{\Lambda}_{\boldsymbol{\Sigma}}\mathbf{W}_{\boldsymbol{\Sigma}}^{T}.
\]
This representation gives us a very useful insight to the noise generated
from \emph{$\mathcal{MVG}_{m,n}(\mathbf{0},\boldsymbol{\Sigma},\boldsymbol{\Psi})$}:
it tells us the \emph{directions} of the noise via the column vectors
of $\mathbf{W}_{\boldsymbol{\Sigma}}$, and \emph{variance} of the
noise in each direction via the singular values in $\boldsymbol{\Lambda}_{\boldsymbol{\Sigma}}$.

For simplicity, consider a two-dimensional multivariate Gaussian distribution,
i.e. $m=2$, so there are two column vectors of $\mathbf{W}_{\boldsymbol{\Sigma}}=[\mathbf{w}_{\boldsymbol{\Sigma}1},\mathbf{w}_{\boldsymbol{\Sigma}2}]$.
The geometry of this distribution can be depicted by an ellipsoid,
e.g. the dash contour in Fig. \ref{fig:dir_noise_samples}, Left (cf.
\cite[ch. 4]{RefWorks:51}, \cite[ch. 2]{RefWorks:225}). This ellipsoid
is characterized by its two principal axes \textendash{} the major
and the minor axes. It is well-known that the two column vectors from
SVD, i.e. $\mathbf{w}_{\boldsymbol{\Sigma}1}$ and $\mathbf{w}_{\boldsymbol{\Sigma}2}$,
are unit vectors pointing in the directions of the major and minor
axes of this ellipsoid, and more importantly, the length of each axis
is characterized by its corresponding singular value, i.e. $\sigma_{\boldsymbol{\Sigma}1}$
and $\sigma_{\boldsymbol{\Sigma}2}$, respectively (cf. \cite[ch. 4]{RefWorks:51})
(recall from Theorem \ref{thm:svd} that $diag(\boldsymbol{\Lambda}_{\boldsymbol{\Sigma}})=[\sigma_{\boldsymbol{\Sigma}1},\sigma_{\boldsymbol{\Sigma}2}]$).
This is illustrated by Fig. \ref{fig:dir_noise_samples}, Left. Therefore,
when considering this 2D multivariate Gaussian noise, we arrive at
the following interpretation of the SVD of its covariance matrix:
the noise is distributed toward the two principal \emph{directions}
specified by $\mathbf{w}_{\boldsymbol{\Sigma}1}$ and $\mathbf{w}_{\boldsymbol{\Sigma}2}$,
with the \emph{variance} scaled by the respective singular values,
$\sigma_{\boldsymbol{\Sigma}1}$ and $\sigma_{\boldsymbol{\Sigma}2}$.

We can extend this interpretation to a more general case with $m>2$,
and also to the other covariance matrix $\boldsymbol{\Psi}$. Then,
we have a full interpretation of \emph{$\mathcal{MVG}_{m,n}(\mathbf{0},\boldsymbol{\Sigma},\boldsymbol{\Psi})$}
as follows. The matrix-valued noise distributed according to \emph{$\mathcal{MVG}_{m,n}(\mathbf{0},\boldsymbol{\Sigma},\boldsymbol{\Psi})$}
has two components: the\emph{ row-wise noise}, and the \emph{column-wise
noise}. The row-wise noise and the column-wise noise are characterized
by the two covariance matrices, $\boldsymbol{\Sigma}$ and $\boldsymbol{\Psi}$,
respectively, as follows.

For the \emph{row-wise noise}.
\begin{itemize}
\item The row-wise noise is characterized by $\boldsymbol{\Sigma}$. 
\item SVD of $\boldsymbol{\Sigma}=\mathbf{W}_{\boldsymbol{\Sigma}}\boldsymbol{\Lambda}_{\boldsymbol{\Sigma}}\mathbf{W}_{\boldsymbol{\Sigma}}^{T}$
decomposes the row-wise noise into two components \textendash{} the
directions and the variances of the noise in those directions. 
\item The \emph{directions} of the row-wise noise are specified by the column
vectors of $\mathbf{W}_{\boldsymbol{\Sigma}}$. 
\item The \emph{variance} of each row-wise-noise direction is indicated
by its corresponding singular value in $\boldsymbol{\Lambda}_{\boldsymbol{\Sigma}}$. 
\end{itemize}

For the \emph{column-wise noise}.
\begin{itemize}
\item The column-wise noise is characterized by $\boldsymbol{\Psi}$. 
\item SVD of $\boldsymbol{\Psi}=\mathbf{W}_{\boldsymbol{\Psi}}\boldsymbol{\Lambda}_{\boldsymbol{\Psi}}\mathbf{W}_{\boldsymbol{\Psi}}^{T}$
decomposes the column-wise noise into two components \textendash{}
the directions and the variances of the noise in those directions. 
\item The \emph{directions} of the column-wise noise are specified by the
column vectors of $\mathbf{W}_{\boldsymbol{\Psi}}$. 
\item The \emph{variance} of each column-wise-noise direction is indicated
by its respective singular value in $\boldsymbol{\Lambda}_{\boldsymbol{\Psi}}$. 
\end{itemize}
Since \emph{$\mathcal{MVG}_{m,n}(\mathbf{0},\boldsymbol{\Sigma},\boldsymbol{\Psi})$}
is fully characterized by its covariances, these two components of
the matrix-valued noise drawn from \emph{$\mathcal{MVG}_{m,n}(\mathbf{0},\boldsymbol{\Sigma},\boldsymbol{\Psi})$}
provide a complete interpretation of the matrix-variate Gaussian noise.

\subsection{Directional Noise \& MVG Mechanism}  \label{subsec:Dir_noise_via_mvg}

We now revisit Theorem \ref{thm:design_general}. Recall that the
sufficient $(\epsilon,\delta)$-differential-privacy condition for
the MVG mechanism puts the constraint only on the singular values
of $\boldsymbol{\Sigma}$ and $\boldsymbol{\Psi}$. However, as we
discuss in the previous section, the singular values of $\boldsymbol{\Sigma}$
and $\boldsymbol{\Psi}$ only indicate the \emph{variance} of the
noise in each direction, but \emph{not the directions they are attributed
to}. In other words, Theorem \ref{thm:design_general} suggests that
the MVG mechanism preserves $(\epsilon,\delta)$-differential privacy
\emph{as long as the overall variances of the noise satisfy a certain
threshold, but these variances can be attributed non-uniformly in
any direction.}

This claim certainly warrants further discussion, and we defer it
to Sec. \ref{sec:Practical-Implementation}, where we present the technical
detail on how to practically implement this concept of directional
noise. It is important to only note here that this claim \emph{does
not mean} that we can avoid adding noise in any particular direction
altogether. On the contrary, there is still a minimum amount of noise
\emph{required in every direction} for the MVG mechanism to guarantee
differential privacy, but the noise simply can be attributed unevenly
in different directions (see Fig. \ref{fig:dir_noise_samples}, Right, for an example).

\subsection{Utility Gain via Directional Noise } \label{subsec:Utility-Gain-via-dir-noise}

There are multiple ways to exploit the notion of directional noise
to enhance utility of differential privacy. Here, we present two methods
\textendash{} via the domain knowledge and via the SVD/PCA.

\subsubsection{Utilizing Domain Knowledge}

This method is best described by examples. Consider first the personalized
warfarin dosing problem \cite{RefWorks:404}, which can be considered
as the regression problem with the identity query, $f(\mathbf{X})=\mathbf{X}$.
In the i.i.d. noise scheme, every feature used in the warfarin dosing
prediction is equally perturbed. However, domain experts may have
prior knowledge that some features are more critical than the others,
so adding directional noise designed such that the more critical features
are perturbed less can potentially yield better prediction performance.

Let us next consider a slightly more involved matrix-valued query:
the covariance matrix, i.e. $f(\mathbf{X})=\frac{1}{N}\mathbf{X}\mathbf{X}^{T}$,
where $\mathbf{X}\in\mathbb{R}^{M\times N}$ has zero mean and the
columns are the records/samples. Consider now the Netflix prize dataset
\cite{RefWorks:377,RefWorks:252}. A popular method for solving the
Netflix challenge is via low-rank approximation \cite{RefWorks:300},
which often involves the covariance matrix query function \cite{RefWorks:194,RefWorks:408,RefWorks:178}.
Suppose domain experts indicate that some features are more informative
than the others. Since the covariance matrix has the underlying property
that each row and column correspond to a single feature \cite{RefWorks:51},
we can use this domain knowledge with directional noise by adding
less noise to the rows and columns corresponding to the informative
features. 

In both examples, the directions chosen are among the \emph{standard
basis}, e.g. $\mathbf{v}_{1}=[1,0,\ldots,0]^{T},\mathbf{v}_{2}=[0,1,\ldots,0]^{T}$,
which are ones of the simplest forms of directions.

\subsubsection{Using Differentially-Private SVD/PCA}

When domain knowledge is not available, an alternative approach is
to derive the directions via the SVD or PCA. In this context, SVD
and PCA are identical with the main difference being that SVD is compatible
with any matrix-valued query function, while PCA is best suited to
the identity query. Hence, the terms may be used interchangeable in
the subsequent discussion.

As we show in Sec. \ref{subsec:dir_noise_as_noniid}, SVD/PCA can decompose
a matrix into its directions and variances. Hence, we can set aside
a small portion of privacy budget to derive the directions from the
SVD/PCA of the query function. This is illustrated in the following
example. Consider again the warfarin dosing problem \cite{RefWorks:404},
and assume that we do not possess any prior knowledge about the predictive
features. We can learn this information from the data by spending
a small privacy budget on deriving differentially-private principal
components (P.C.). Each P.C. can then serve as a direction and, with
directional noise, we can selectively add less noise in the highly
informative directions as indicated by PCA. 

Clearly, the directions in this example are not necessary among the
standard basis, but can be \emph{any unit vector}. This example illustrates
how directional noise can provide additional utility benefit even
without the domain knowledge. There have been many works on differentially-private
SVD/PCA \cite{RefWorks:405,RefWorks:406,RefWorks:194,RefWorks:249,RefWorks:178,RefWorks:313,hardt2012beating,hardt2013beyond},
so this method is very generally applicable. Again, we reiterate that
the approach similar to the one in the example using SVD applies to
a general matrix-valued query function. Fig. \ref{fig:dir_noise_samples},
Right, illustrates this. In the illustration, the query function has
two dimensions, and we have obtained the utility direction, e.g. from
SVD, as represented by the green line. This can be considered as the
utility subspace we desire to be least perturbed. The many small circles
in the illustration represent how the i.i.d. noise and directional
noise are distributed under the 2D multivariate Gaussian distribution.
Clearly, directional noise can reduce the perturbation experienced
on the utility directions.

In the next section, we discuss how to implement directional noise
with the MVG mechanism in practice and propose two simple algorithms
for two types of directional noise.

\section{Practical Implementation} \label{sec:Practical-Implementation}

The differential privacy condition in Theorem \ref{thm:design_general},
even along with the notion of directional noise in the previous section,
still leads to a large design space for the MVG mechanism. In this
section, we present two simple algorithms to implement the MVG mechanism
with two types of directional noise that can be appropriate for a
wide range of real-world applications. Then, we conclude the section
with a discussion on a sampling algorithm for \emph{$\mathcal{MVG}_{m,n}(\mathbf{0},\boldsymbol{\Sigma},\boldsymbol{\Psi})$}.

As discussed in Sec. \ref{subsec:Dir_noise_via_mvg}, Theorem \ref{thm:design_general}
states that the MVG mechanism satisfies $(\epsilon,\delta)$-differential
privacy as long as the singular values of $\boldsymbol{\Sigma}$ and
$\boldsymbol{\Psi}$ satisfy the sufficient condition. This provides
tremendous flexibility in the choice of the \emph{directions of the
noise}. First, we notice from the sufficient condition in Theorem
\ref{thm:design_general} that the singular values for $\boldsymbol{\Sigma}$
and $\boldsymbol{\Psi}$ are decoupled, i.e. they can be designed
independently so long as, when combined, they satisfy the specified
condition. Hence, the \emph{row-wise noise} and \emph{column-wise
noise} can be considered as the \emph{two modes of noise} in the MVG
mechanism. By this terminology, we discuss two types of directional
noise: the \emph{unimodal} and \emph{equi-modal} directional noise.

\subsection{Unimodal Directional Noise}  \label{subsec:Unimodal-Directional-Noise}

For the unimodal directional noise, we select \emph{one mode} of the
noise to be directional noise, whereas the other mode of the noise
is set to be i.i.d. For this discussion, we assume that the row-wise
noise is directional noise, while the column-wise noise is i.i.d.
However, the opposite case can be readily analyzed with the similar
analysis.

We note that, apart from simplifying the practical implementation
that we will discuss shortly, this type of directional noise can be
appropriate for many applications. For example, for the identity query,
we may not possess any prior knowledge on the quality of each sample,
so the best strategy would be to consider the i.i.d. column-wise noise
(recall that samples are the column vectors).

Formally, the unimodal directional noise sets $\boldsymbol{\Psi}=\mathbf{I}_{n}$,
where $\mathbf{I}_{n}$ is the $n\times n$ identity matrix. This,
consequently, reduces the design space for the MVG mechanism with
directional noise to only that of $\boldsymbol{\Sigma}$. Next, consider
the left side of Eq. \eqref{eq.sufficient_condition}, and we have
\begin{equation}
\left\Vert \boldsymbol{\sigma}(\boldsymbol{\Sigma}^{-1})\right\Vert _{2}=\sqrt{\sum_{i=1}^{m}\frac{1}{\sigma_{i}^{2}(\boldsymbol{\Sigma})}}\ \textrm{, and}\ \left\Vert \boldsymbol{\sigma}(\boldsymbol{\Psi}^{-1})\right\Vert _{2}=\sqrt{n}.\label{eq:singular_u}
\end{equation}
If we square both sides of the sufficient condition and re-arrange
it, we get a form of the condition such that the row-wise noise in
each direction is decoupled: 
\begin{equation}
\sum_{i=1}^{m}\frac{1}{\sigma_{i}^{2}(\boldsymbol{\Sigma})}\leq\frac{1}{n}\frac{(-\beta+\sqrt{\beta^{2}+8\alpha\epsilon})^{4}}{16\alpha^{4}}.\label{eq:decoupled_condition}
\end{equation}

This form gives a very intuitive interpretation of the directional
noise. First, we note that, to have small noise in the $i^{th}$ direction,
$\sigma_{i}(\boldsymbol{\Sigma})$ has to be small (cf. Sec. \ref{subsec:dir_noise_as_noniid}).
However, the sum of $1/\sigma_{i}^{2}(\boldsymbol{\Sigma})$ of the
noise in all directions, which should hence be large, is limited by
the quantity on the right side of Eq. (\ref{eq:decoupled_condition}).
This, in fact, explains why even with directional noise, we still
need to add noise in \emph{every direction} to guarantee differential
privacy. Consider the case when we set the noise in one direction
to be zero, and we have $\underset{\sigma\rightarrow0}{\lim}\frac{1}{\sigma}=\infty$,
which immediately violates the sufficient condition in Eq. (\ref{eq:decoupled_condition}).

From Eq. (\ref{eq:decoupled_condition}), the quantity $1/\sigma_{i}^{2}(\boldsymbol{\Sigma})$
is the inverse of the variance of the noise in the $i^{th}$ direction,
so we may think of it as the \emph{precision} measure of the query
answer in that direction. The intuition is that the higher this value
is, the lower the noise added in that direction, and, hence, the more
precise the query value in that direction is. From this description,
the constraint in Eq. (\ref{eq:decoupled_condition}) can be aptly
named as the \emph{precision budget}, and we immediately have the
following theorem. 
\begin{thm}
For the MVG mechanism with $\boldsymbol{\Psi}=\mathbf{I}_{n}$, the
precision budget is $(-\beta+\sqrt{\beta^{2}+8\alpha\epsilon})^{4}/(16\alpha^{4}n)$. 
\end{thm}
Finally, the remaining task is to determine the directions of the
noise and form $\boldsymbol{\Sigma}$ accordingly. To do so systematically,
we first decompose $\boldsymbol{\Sigma}$ by SVD as, 
\[
\boldsymbol{\Sigma}=\mathbf{W}_{\boldsymbol{\Sigma}}\boldsymbol{\Lambda}_{\boldsymbol{\Sigma}}\mathbf{W}_{\boldsymbol{\Sigma}}^{T}.
\]
This decomposition represents $\boldsymbol{\Sigma}$ by two components
\textendash{} the directions of the row-wise noise indicated by $\mathbf{W}_{\boldsymbol{\Sigma}}$,
and the variance of the noise indicated by $\boldsymbol{\Lambda}_{\boldsymbol{\Sigma}}$.
Since the precision budget only puts constraint upon $\boldsymbol{\Lambda}_{\boldsymbol{\Sigma}}$,
this decomposition allows us to freely chose any unitary matrix for
$\mathbf{W}_{\boldsymbol{\Sigma}}$ such that each column of $\mathbf{W}_{\boldsymbol{\Sigma}}$
indicates each independent direction of the noise.

Therefore, we present the following simple approach to design the
MVG mechanism with the unimodal directional noise: under a given precision
budget, allocate more precision to the directions of more importance.

Alg. \ref{alg:mvg_design_unimodal} formalizes this procedure. It
takes as inputs, among other parameters, the \emph{precision allocation
strategy} $\boldsymbol{\theta}\in(0,1)^{m}$, and the \emph{directions}
$\mathbf{W}_{\boldsymbol{\Sigma}}\in\mathbb{R}^{m\times m}$. The
precision allocation strategy is a vector of size $m$, whose elements,
$\theta_{i}\in(0,1)$, corresponds to the importance of the $i^{th}$
direction indicated by the $i^{th}$ orthonormal column vector of
$\mathbf{W}_{\boldsymbol{\Sigma}}$. The higher the value of $\theta_{i}$,
the more important the $i^{th}$ direction is. Moreover, the algorithm
enforces that $\sum_{i=1}^{m}\theta_{i}\leq1$ to ensure that the
precision budget is not overspent. The algorithm proceeds as follows.
First, compute $\alpha$ and $\beta$ and, then, the precision budget
$P$. Second, assign precision to each direction based on the precision
allocation strategy. Third, derive the variance of the noise in each
direction accordingly. Then, compute $\boldsymbol{\Sigma}$ from the
noise variance and directions, and draw a matrix-valued noise from
$\mathcal{MVG}_{m,n}(\mathbf{0},\boldsymbol{\Sigma},\mathbf{I})$.
Finally, output the query answer with additive matrix noise.

We make a remark here about choosing directions of the noise. As discussed
in Sec. \ref{sec:Directional-Noise}, any orthonormal set of vectors can
be used as the directions. The simplest instance is the
the standard basis vectors, e.g. $\mathbf{e}_{1}=[1,0,0]^{T},\mathbf{e}_{2}=[0,1,0]^{T},\mathbf{e}_{3}=[0,0,1]^{T}$
for $\mathbb{R}^{3}$.

\begin{algorithm}
\begin{flushleft}
\textbf{Input:}{ (a) privacy parameters: $\epsilon,\delta$;
(b) the query function and its sensitivity: $f(\mathbf{X})\in\mathbb{R}^{m\times n},s_{2}(f)$;
(c) the precision allocation strategy $\boldsymbol{\theta}\in(0,1)^{m}:\left|\boldsymbol{\theta}\right|_{1}=1$;
and (d) the $m$ directions of the row-wise noise $\mathbf{W}_{\boldsymbol{\Sigma}}\in\mathbb{R}^{m\times m}$.}

\vspace{0.5em}

\ (1) Compute $\alpha$ and $\beta$ (cf. Theorem \ref{thm:design_general}).

\ (2) Compute the precision budget $P=\frac{(-\beta+\sqrt{\beta^{2}+8\alpha\epsilon})^{4}}{16\alpha^{4}n}$.

\ (3) \textbf{for}{ $i=1,\ldots,m$: }

\vspace{0.5em}

{\quad{}i) Set $p_{i}=\theta_{i}P$.}

{\quad{}ii) Compute the $i^{th}$ direction's variance, $\sigma_{i}(\boldsymbol{\Sigma})=1/\sqrt{p_{i}}$.}

\vspace{0.5em}

\ {(4) Form the diagonal matrix $\boldsymbol{\Lambda}_{\boldsymbol{\Sigma}}=diag([\sigma_{1}(\boldsymbol{\Sigma}),\ldots,\sigma_{m}(\boldsymbol{\Sigma})])$.}

\ {(5) Derive the covariance matrix: $\boldsymbol{\Sigma}=\mathbf{W}_{\boldsymbol{\Sigma}}\boldsymbol{\Lambda}_{\boldsymbol{\Sigma}}\mathbf{W}_{\boldsymbol{\Sigma}}^{T}$.}

\ {(6) Draw a matrix-valued noise $\mathcal{Z}$ from $\mathcal{MVG}_{m,n}(\mathbf{0},\boldsymbol{\Sigma},\mathbf{I})$.}

\vspace{0.5em}
\textbf{Output:}{ $f(\mathbf{X})+\mathcal{Z}$.}

\end{flushleft} 

\caption{MVG mech. w/ unimodal directional noise.} \label{alg:mvg_design_unimodal}
\end{algorithm}

\subsection{Equi-Modal Directional Noise}  \label{subsec:Equi-Modal-Directional-Noise}

Next, we consider the type of directional noise of which the row-wise
noise and column-wise noise are distributed identically, which we
call the equi-modal directional noise. We recommend this type of directional
noise for a \emph{symmetric query function}, i.e. $f(\mathbf{X})=f(\mathbf{X})^{T}\in\mathbb{R}^{m\times m}$.
This covers a wide-range of query functions including the covariance
matrix \cite{RefWorks:249,RefWorks:194,RefWorks:178}, the kernel
matrix \cite{RefWorks:33}, the adjacency matrix of an undirected
graph \cite{RefWorks:329}, and the Laplacian matrix \cite{RefWorks:329}.
The motivation for this recommendation is that, for symmetric query
functions, any prior information about the rows would similarly apply
to the columns, so it is reasonable to use identical row-wise and
column-wise noise.

Formally, this type of directional noise imposes that $\boldsymbol{\Psi}=\boldsymbol{\Sigma}$.
Following a similar derivation to the unimodal type, we have the following
precision budget. 
\begin{thm}
\label{thm:equi_budget} For the MVG mechanism with $\boldsymbol{\Psi}=\boldsymbol{\Sigma}$,
the precision budget is $(-\beta+\sqrt{\beta^{2}+8\alpha\epsilon})^{2}/(4\alpha^{2})$. 
\end{thm}
Following a similar procedure to the unimodal type, we present Alg.
\ref{alg:mvg_design_equimodal} for the MVG mechanism with the equi-modal
directional noise. The algorithm follows the same steps as Alg. \ref{alg:mvg_design_unimodal},
except it derives the precision budget from Theorem \ref{thm:equi_budget},
and draws the noise from $\mathcal{MVG}_{m,m}(\mathbf{0},\boldsymbol{\Sigma},\boldsymbol{\Sigma})$.

\begin{algorithm}
\begin{flushleft}
\textbf{Input:}{ (a) privacy parameters: $\epsilon,\delta$;
(b) the query function and its sensitivity: $f(\mathbf{X})\in\mathbb{R}^{m\times m},s_{2}(f)$;
(c) the precision allocation strategy $\boldsymbol{\theta}\in(0,1)^{m}:\left|\boldsymbol{\theta}\right|_{1}=1$;
and (d) the $m$ noise directions $\mathbf{W}_{\boldsymbol{\Sigma}}\in\mathbb{R}^{m\times m}$.}

\vspace{0.5em}
\ {(1) Compute $\alpha$ and $\beta$ (cf. Theorem \ref{thm:design_general}).}

\ {(2) Compute the precision budget $P=\frac{(-\beta+\sqrt{\beta^{2}+8\alpha\epsilon})^{2}}{4\alpha^{2}}$.}

\ (3) \textbf{for}{ $i=1,\ldots,m$: }

\vspace{0.5em}
{\quad{}}{i) Set $p_{i}=\theta_{i}P$.}

{\quad{}}{ii) Compute the the $i^{th}$ direction's
variance, $\sigma_{i}(\boldsymbol{\Sigma})=1/\sqrt{p_{i}}$.}
\vspace{0.5em}

\ {(4) Form the diagonal matrix $\boldsymbol{\Lambda}_{\boldsymbol{\Sigma}}=diag([\sigma_{1}(\boldsymbol{\Sigma}),\ldots,\sigma_{m}(\boldsymbol{\Sigma})])$.}

\ {(5) Derive the covariance matrix: $\boldsymbol{\Sigma}=\mathbf{W}_{\boldsymbol{\Sigma}}\boldsymbol{\Lambda}_{\boldsymbol{\Sigma}}\mathbf{W}_{\boldsymbol{\Sigma}}^{T}$.}

\ {(6) Draw a matrix-valued noise $\mathcal{Z}$ from $\mathcal{MVG}_{m,m}(\mathbf{0},\boldsymbol{\Sigma},\boldsymbol{\Sigma})$.}

\vspace{0.5em}
\textbf{Output:}{ $f(\mathbf{X})+\mathcal{Z}$.}

\end{flushleft}

\caption{MVG mech. w/ equi-modal directional noise.} \label{alg:mvg_design_equimodal}
\end{algorithm}

\subsection{Sampling from \texorpdfstring{$\mathcal{MVG}_{m,n}(\mathbf{0},\boldsymbol{\Sigma},\boldsymbol{\Psi})$}{MVG(0,Sigma,Psi)} }
\label{subsec:Sampling-from-mvg}

One remaining question on the practical implementation of the MVG
mechanism is how to efficiently draw the noise from the matrix-variate Gaussian distribution $\mathcal{MVG}_{m,n}(\mathbf{0},\boldsymbol{\Sigma},\boldsymbol{\Psi})$.
One approach to implement a sampler for $\mathcal{MVG}_{m,n}(\mathbf{0},\boldsymbol{\Sigma},\boldsymbol{\Psi})$
is via the affine transformation of samples drawn i.i.d. from the
standard normal distribution, i.e. $\mathcal{N}(0,1)$. The transformation
is described by the following lemma \cite{RefWorks:279}. 
\begin{lem}
\label{lem:affine_tx}Let $\mathcal{N}\in\mathbb{R}^{m\times n}$
be a matrix-valued random variable whose elements are drawn i.i.d.
from the standard normal distribution $\mathcal{N}(0,1)$. Then, the
matrix $\mathcal{Z}=\mathbf{B}_{\boldsymbol{\Sigma}}\mathcal{N}\mathbf{B}_{\boldsymbol{\Psi}}^{T}$
is distributed according to $\mathcal{Z}\sim\mathcal{MVG}_{m,n}(\mathbf{0},\mathbf{B}_{\boldsymbol{\Sigma}}\mathbf{B}_{\boldsymbol{\Sigma}}^{T},\mathbf{B}_{\boldsymbol{\Psi}}\mathbf{B}_{\boldsymbol{\Psi}}^{T})$. 
\end{lem}
This transformation consequently allows the conversion between $mn$
samples drawn i.i.d. from $\mathcal{N}(0,1)$ and a sample drawn from
$\mathcal{MVG}_{m,n}(\mathbf{0},\mathbf{B}_{\boldsymbol{\Sigma}}\mathbf{B}_{\boldsymbol{\Sigma}}^{T},\mathbf{B}_{\boldsymbol{\Psi}}\mathbf{B}_{\boldsymbol{\Psi}}^{T})$.
To derive $\mathbf{B}_{\boldsymbol{\Sigma}}$ and $\mathbf{B}_{\boldsymbol{\Psi}}$
from given $\boldsymbol{\Sigma}$ and $\boldsymbol{\Psi}$ for $\mathcal{MVG}_{m,n}(\mathbf{0},\boldsymbol{\Sigma},\boldsymbol{\Psi})$,
we solve the two linear equations: $\mathbf{B}_{\boldsymbol{\Sigma}}\mathbf{B}_{\boldsymbol{\Sigma}}^{T}=\boldsymbol{\Sigma}$,
and $\mathbf{B}_{\boldsymbol{\Psi}}\mathbf{B}_{\boldsymbol{\Psi}}^{T}=\boldsymbol{\Psi}$,
and the solutions of these two equations can be acquired readily via
the Cholesky decomposition or SVD (cf. \cite{RefWorks:208}). We summarize
the steps for this implementation here using SVD: 
\begin{enumerate}
\item Draw $mn$ i.i.d. samples from $\mathcal{N}(0,1)$, and form a matrix
$\mathcal{N}$. 
\item Let $\mathbf{B}_{\boldsymbol{\Sigma}}=\mathbf{W}_{\boldsymbol{\Sigma}}\boldsymbol{\Lambda}_{\boldsymbol{\Sigma}}^{1/2}$
and $\mathbf{B}_{\boldsymbol{\Psi}}=\mathbf{W}_{\boldsymbol{\Psi}}\boldsymbol{\Lambda}_{\boldsymbol{\Psi}}^{1/2}$,
where $\mathbf{W}_{\boldsymbol{\Sigma}},\boldsymbol{\Lambda}_{\boldsymbol{\Sigma}}$
and $\mathbf{W}_{\boldsymbol{\Psi}},\boldsymbol{\Lambda}_{\boldsymbol{\Psi}}$
are derived from SVD of $\boldsymbol{\Sigma}$ and $\boldsymbol{\Psi}$,
respectively. 
\item Compute the sample $\mathcal{Z}=\mathbf{B}_{\boldsymbol{\Sigma}}\mathcal{N}\mathbf{B}_{\boldsymbol{\Psi}}^{T}$. 
\end{enumerate}
The complexity of this method depends on that of the $\mathcal{N}(0,1)$
sampler used. Plus, there is an additional $\mathcal{O}(\max\{m^{3},n^{3}\})$
complexity from SVD \cite{RefWorks:451}\footnote{Note that $n$ here is \emph{not} the number of samples or records but is the dimension of the matrix-valued query output, i.e. $f(\mathbf{X})\in \mathbb{R}^{m\times n}$.}. The memory needed is in the order of $m^2+n^2+mn$ from the three matrices required in step (3).

\section{Experimental Setups}

\label{sec:Experiments}

We evaluate the proposed MVG mechanism on three experimental setups
and datasets. Table \ref{tab:exp_setups} summarizes our setups. In
all experiments, 100 trials are carried out and the average and 95\%
confidence interval are reported.

\subsection{Experiment I: Regression}

\subsubsection{Task and Dataset.}

The first experiment considers the regression application on the Liver
Disorders dataset \cite{RefWorks:322,RefWorks:410}, which contains
5 features from the blood sample of 345 patients. We leave out the
samples from 97 patients for testing, so the private dataset contains
248 patients. Following suggestions by Forsyth and Rada \cite{RefWorks:413},
we use these features to predict the average daily alcohol consumption.
All features and teacher values are $\in[0,1]$.

\subsubsection{Query Function and Evaluation Metric.}

We perform regression in a differentially-private manner via the identity
query, i.e. $f(\mathbf{X})=\mathbf{X}$. Since regression involves
the teacher values, we treat them as a feature, so the query size
becomes $6\times248$. We use the kernel ridge regression (KRR) \cite{RefWorks:33,RefWorks:231}
as the regressor, and the root-mean-square error (RMSE) \cite{RefWorks:51,RefWorks:33}
as the evaluation metric.

\subsubsection{MVG Mechanism Design.}

As discussed in Sec. \ref{subsec:Unimodal-Directional-Noise}, Alg. \ref{alg:mvg_design_unimodal}
is appropriate for the identity query, so we employ it for this experiment.
The $L_{2}$-sensitivity of this query is $\sqrt{6}$ (cf. Appendix
\ref{sec:L2-Sensitivities}). To identify the informative directions
to allocate the precision budget, we implement both methods discussed
in Sec. \ref{subsec:Utility-Gain-via-dir-noise} as follows.

(a) For the method using
domain knowledge (denoted \emph{MVG-1}), we refer to Alatalo et al.
\cite{RefWorks:395}, which indicates that alanine aminotransferase
(ALT) is the most indicative feature for predicting the alcohol consumption
behavior. Additionally, from our prior experience working with regression
problems, we anticipate that the teacher value (Y) is another important
feature to allocate more precision budget to. With this setup, we
use the standard basis vectors as the directions (cf. Sec. \ref{subsec:Utility-Gain-via-dir-noise}),
and employ the following \emph{binary precision allocation strategy}. 
\begin{itemize}
\item Allocate $\tau$\% of the precision budget to the two
important features (ALT and Y) by equal amount. 
\item Allocate the rest of the precision budget equally to the rest of the
features. 
\end{itemize}
We vary $\tau\in\{55,65,\ldots,95\}$ and report the best results.\footnote{In the real-world deployment, this parameter selection process should
also be made private \cite{RefWorks:417}.} 

(b) For the method using differentially-private SVD/PCA (denoted \emph{MVG-2}), given the total budget of $\{\epsilon,\delta\}$ reported in Sec. \ref{sec:experimental_results}, we spend
 $0.2\epsilon$ and $0.2\delta$  on the derivation of
the two most informative directions via the differentially-private
PCA algorithm in \cite{RefWorks:194}. We specify the first two principal components
as the indicative features for a fair comparison with the method using
domain knowledge. The remaining $0.8\epsilon$ and $0.8\delta$ are then used for Alg. \ref{alg:mvg_design_unimodal}. Again, for a fair comparison with the method using domain knowledge, we use the same binary precision
allocation strategy in Alg. \ref{alg:mvg_design_unimodal} for \emph{MVG-2}.

\subsection{Experiment II: \texorpdfstring{1$^{st}$}{1st} Principal Component}

\begin{table}
\begin{centering}
\begin{tabular}{|>{\centering}p{1.6cm}|>{\centering}p{1.76cm}|>{\centering}p{1.8cm}|>{\centering}p{1.76cm}|}
\hline 
 & {Exp. I } & {Exp. II } & {Exp. III}\tabularnewline
\hline 
\hline 
{Task } & {Regression } & {$1^{st}$ P.C. } & {Covariance estimation}\tabularnewline
\hline 
{Dataset } & {Liver \cite{RefWorks:322,RefWorks:413} } & {Movement \cite{RefWorks:396} } & {CTG \cite{RefWorks:322,RefWorks:412}}\tabularnewline
\hline 
{\# samples $N$} & {248} & {2,021} & {2,126}\tabularnewline
\hline 
{\# features $M$} & {6} & {4} & {21}\tabularnewline
\hline 
{Query $f(\mathbf{X})$ } & {$\mathbf{X}$ } & {$\mathbf{X}\mathbf{X}^{T}/N$ } & {$\mathbf{X}$}\tabularnewline
\hline 
{Query size } & {$6\times248$ } & {$4\times4$ } & {$21\times2126$}\tabularnewline
\hline 
{Eval. metric } & {RMSE } & {$\Delta\rho$ (Eq. (\ref{eq:delta_rho})) } & {RSS (Eq. (\ref{eq:rss}))}\tabularnewline
\hline 
{MVG Alg. } & {1 } & {2 } & {1}\tabularnewline
\hline 
{Source of directions } & {Domain knowledge \cite{RefWorks:395} /PCA \cite{RefWorks:194} } & {Data collection setup \cite{RefWorks:396}} & {Domain knowledge \cite{RefWorks:416}}\tabularnewline
\hline 
\end{tabular}
\par\end{centering}
\caption{The three experimental setups.} \label{tab:exp_setups}
 
\end{table}

\subsubsection{Task and Dataset.}

The second experiment considers the problem of determining the first
principal component ($1^{st}$ P.C.) from the principal component
analysis (PCA). This is one of the most popular problems in machine
learning and differential privacy. We only consider the first principal
component for two reasons. First, many prior works in differentially-private
PCA algorithm consider this problem or the similar problem of deriving
a few major P.C. (cf. \cite{RefWorks:178,RefWorks:313,RefWorks:249}),
so this allows us to compare our approach to the state-of-the-art
approaches of a well-studied problem. Second, in practice, this method
for deriving the $1^{st}$ P.C. may be used iteratively to derive
the rest of the principal components (cf. \cite{RefWorks:414}).

We use the Movement Prediction via RSS (Movement) dataset \cite{RefWorks:396},
which consists of the radio signal strength measurement from 4 sensor
anchors (ANC\{0-3\}) \textendash{} corresponding to the 4 features
\textendash{} from 2,021 movement samples. The feature data all have
the range of $[-1,1]$.

\subsubsection{Query Function and Evaluation Metric.}

We consider the covariance matrix query, i.e. $f(\mathbf{X})=\frac{1}{N}\mathbf{X}\mathbf{X}^{T}$,
and use SVD to derive the $1^{st}$ P.C. from it. Hence, the query
size is $4\times4$. We adopt the quality metric commonly used for
P.C. \cite{RefWorks:33} and also used by Dwork et al. \cite{RefWorks:249},
i.e. the \emph{captured variance} $\rho$. For a given P.C. $\mathbf{v}$,
the capture variance by $\mathbf{v}$ on the covariance matrix $\bar{\mathbf{S}}$
is defined as $\rho=\mathbf{v}^{T}\bar{\mathbf{S}}\mathbf{v}$. To
be consistent with other experiments, we report the absolute error
in $\rho$ as deviated from the maximum $\rho$. It is well-established
that the maximum $\rho$ is equal to the largest eigenvalue of $\bar{\mathbf{S}}$
(cf. \cite[Theorem 4.2.2]{RefWorks:208}, \cite{RefWorks:415}). Hence,
the metric can be written as, 
\begin{equation}
\Delta\rho(\mathbf{v})=\lambda_{1}-\rho(\mathbf{v}),\label{eq:delta_rho}
\end{equation}
where $\lambda_{1}$ is the largest eigenvalue of $\bar{\mathbf{S}}$.
For the ideal, non-private case, the error would clearly be zero.

\subsubsection{MVG Mechanism Design.}

As discussed in Sec. \ref{subsec:Equi-Modal-Directional-Noise}, Alg.
\ref{alg:mvg_design_equimodal} is appropriate for the covariance
query, so we employ it for this experiment. The $L_{2}$-sensitivity
of this query is $8/2021$ (cf. Appendix \ref{sec:L2-Sensitivities}).
To identify the informative directions to allocate the precision budget,
we inspect the data collection setup described in \cite{RefWorks:396},
and use two of the four anchors as the more informative anchors due
to their proximity to the movement path (ANC0 and ANC3). Hence, we
use the standard basis vectors as the directions (cf. Sec. \ref{sec:Practical-Implementation})
and allocate more precision budget to these two features using the
same strategy as in Exp. I.

\subsection{Experiment III: Covariance Estimation}

\subsubsection{Task and Dataset.}

The third experiment considers the similar problem to Exp. II but
with a different flavor. In this experiment, we consider the task
of estimating the covariance matrix \emph{from the perturbed database}.
This differs from Exp. II in three ways. First, for covariance estimation,
we are interested in every P.C. Second, as mentioned in Exp. II, many
previous works do not consider every P.C., so the previous works for
comparison are different. Third, to give a different taste of our
approach, we consider the method of input perturbation for estimating
the covariance, i.e. query the noisy database and use it to compute
the covariance. We use the Cardiotocography (CTG) dataset \cite{RefWorks:322,RefWorks:412},
which consists of 21 features in the range of $[0,1]$ from 2,126
fetal samples.

\subsubsection{Query Function and Evaluation Metric.}

We consider the method via input perturbation, so we use the identity
query, i.e. $f(\mathbf{X})=\mathbf{X}$. The query size is $21\times2126$.
We adopt the captured variance as the quality metric similar to Exp.
II, but since we are interested in every P.C., we consider the \emph{residual
sum of square (RSS)} \cite{RefWorks:350} of every P.C. This is similar
to the total residual variance used by Dwork et al. (\cite[p. 5]{RefWorks:249}).
Formally, given the perturbed database $\tilde{\mathbf{X}}$, the
covariance estimate is $\tilde{\mathbf{S}}=\frac{1}{N}\tilde{\mathbf{X}}\tilde{\mathbf{X}}^{T}$.
Let $\{\tilde{\mathbf{v}}_{i}\}$ be the set of P.C.'s derived from
$\tilde{\mathbf{S}}$, and the RSS is, 
\begin{equation}
RSS(\tilde{\mathbf{S}})=\sum_{i}(\lambda_{i}-\rho(\tilde{\mathbf{v}}_{i}))^{2},\label{eq:rss}
\end{equation}
where $\lambda_{i}$ is the $i^{th}$ eigenvalue of $\bar{\mathbf{S}}$
(cf. Exp. II), and $\rho(\tilde{\mathbf{v}}_{i})$ is the captured
variance of the $i^{th}$ P.C. derived from $\tilde{\mathbf{S}}$.
Clearly, in the non-private case, $RSS(\bar{\mathbf{S}})=0$.

\subsubsection{MVG Mechanism Design.}

Since we consider the identity query, we employ Alg. \ref{alg:mvg_design_unimodal}
for this experiment. The $L_{2}$-sensitivity is $\sqrt{21}$ (cf.
Appendix \ref{sec:L2-Sensitivities}). To identify the informative
directions to allocate the precision budget to, we refer to the domain
knowledge from Costa Santos et al. \cite{RefWorks:416}, which identifies
three features to be most informative, viz. fetal heart rate (FHR),
\%time with abnormal short term variability (ASV), and \%time with abnormal long term variability (ALV). Hence, we use the standard basis vectors as the directions and
allocate more precision budget to these three features using the same
strategy as in Exp. I.

\subsection{Comparison to Previous Works}

Since our approach falls into the category of basic mechanism, we
compare our work to the four prior state-of-the-art basic mechanisms discussed
in Sec. \ref{subsec:Basic-Mechanisms}, namely, the Laplace mechanism,
the Gaussian mechanism, the Exponential mechanism, and the JL transform
method.

For Exp. I and III, since we consider the identity query, the four
previous works for comparison are by Dwork et al. \cite{RefWorks:195},
Dwork et al. \cite{RefWorks:186}, Blum et al. \cite{RefWorks:174},
and Upadhyay \cite{RefWorks:399}, for the four basic mechanisms, respectively.

For Exp. II, we consider the $1^{st}$ P.C. As this problem has been
well-investigated, we compare our approach to the state-of-the-art
algorithms specially designed for this problem. These four algorithms using
the four prior basic mechanisms are, respectively: Dwork et al. \cite{RefWorks:195},
Dwork et al. \cite{RefWorks:249}, Chaudhuri et al. \cite{RefWorks:178},
and Blocki et al. \cite{RefWorks:313}. We note that these four algorithms chosen for comparison are designed and optimized specifically  for the particular application, so they utilize the positive-semidefinite (PSD) nature of the matrix query. On the other hand, the MVG mechanism used here is generally applicable for matrix queries even beyond the particular application, and makes no assumptions about the PSD structure of the matrix query. In other words, we intentionally give a favorable edge to the compared methods to show that, despite the handicap, the MVG mechanism can still perform comparably well.

For all previous works, we use the parameter values as suggested by
the authors of the method, and vary the free variable before reporting
the best performance.

Finally, we recognize that some of these prior works have a different
privacy guarantee from ours, namely, $\epsilon$-differential privacy.
Nevertheless, we present these prior works for comprehensive coverage
of prior basic mechanisms, and we will keep this difference in mind
when discussing the results.

\section{Experimental Results}
\label{sec:experimental_results}

Table \ref{tab:Ex1_results}, Table \ref{tab:Exp2_results}, and Table \ref{tab:Exp3_results} report the experimental results for Experiment I, II, and III, respectively. The performance shown is an average over 100 trials, along with the 95\% confidence interval.

\begin{table}
\begin{centering}
\begin{tabular}{|c|c|c|c|}
\hline 
{Method } & {$\epsilon$ } & {$\delta$ } & {RMSE ($\times10^{-2}$)}\tabularnewline
\hline 
\hline 
{Non-private } & {- } & {- } & {1.226}\tabularnewline
\hline 
{Random guess } & {- } & {- } & {$\sim3.989$}\tabularnewline
\hline 
{MVG-1 (Alg. \ref{alg:mvg_design_unimodal} + knowledge in \cite{RefWorks:395})} & {1.} & {$1/N$ } & {$1.624\pm0.026$}\tabularnewline
\hline 
{MVG-2 (Alg. \ref{alg:mvg_design_unimodal} + DP-PCA \cite{RefWorks:194})} & {1.} & {$1/N$ } & {$1.643\pm0.023$}\tabularnewline
\hline 
{Gaussian (Dwork et al. \cite{RefWorks:186}) } & {1.} & {$1/N$ } & {$1.913\pm0.069$}\tabularnewline
\hline 
{JL transform (Upadhyay \cite{RefWorks:399}) } & {1.} & {$1/N$ } & {$1.682\pm0.015$}\tabularnewline
\hline 
{Laplace (Dwork et al. \cite{RefWorks:195}) } & {1.} & {0 } & {$2.482\pm0.189$}\tabularnewline
\hline 
{Exponential (Blum et al. \cite{RefWorks:174}) } & {1.} & {0 } & {$2.202\pm0.721$}\tabularnewline
\hline 
\end{tabular}
\par\end{centering}
\caption{Results from Exp. I: regression. MVG-1 derives noise directions from
domain knowledge, while MVG-2 derives them from differentially-private
PCA in \cite{RefWorks:194}. } \label{tab:Ex1_results}
 
\end{table}

\subsection{Experiment I: Regression}

Table \ref{tab:Ex1_results} reports the results for Exp. I. Here
are the key observations. 
\begin{itemize}
\item Compared to the non-private baseline, the best MVG mechanism
(MVG-1) yields similar performance (difference of .004 in RMSE). 
\item Compared to other $(\epsilon,\delta)$-basic mechanisms, i.e. the Gaussian mechanism
and the JL transform, the best MVG mechanism (MVG-1) has better utility (by .003
and .0006 in RMSE, respectively) with the same privacy guarantee. 
\item Compared to other $\epsilon$-basic mechanisms, i.e. the Laplace and Exponential
mechanisms, the best MVG mechanism (MVG-1) provides \emph{significantly} better
utility (\url{~}150\%) with the slightly weaker $(\epsilon,1/N)$-differential
privacy guarantee. 
\item Even when some privacy budget is spent on deriving the direction via PCA\cite{RefWorks:194}
(MVG-2), the MVG mechanism still yields the best performance among
all other non-MVG methods.
\end{itemize}
Overall, the results from regression show the promise of the MVG mechanism.
Our approach can outperform all other $(\epsilon,\delta)$-basic mechanisms.
Although it provides a weaker privacy guarantee than other $\epsilon$-basic
mechanisms, it can provide considerably more utility (\url{~}150\%).
As advocated by Duchi et al. \cite{RefWorks:419} and Fienberg et
al. \cite{RefWorks:418}, this trade-off could be attractive in some
settings, e.g. critical medical or emergency situations.

\subsection{Experiment II: \texorpdfstring{1$^{st}$}{1st} Principal Component}

\begin{table}
\begin{centering}
\begin{tabular}{|c|c|c|c|}
\hline 
{Method } & {$\epsilon$ } & {$\delta$ } & {Err. $\Delta\rho$ ($\cdot10^{-1})$}\tabularnewline
\hline 
\hline 
{Non-private } & {- } & {- } & {$0$}\tabularnewline
\hline 
{Random guess } & {- } & {- } & {$\sim3.671$}\tabularnewline
\hline 
{MVG (Alg. \ref{alg:mvg_design_equimodal}) } & {1.} & {$1/N$ } & {$2.384\pm0.215$}\tabularnewline
\hline 
{Gaussian (Dwork et al. \cite{RefWorks:249}) } & {1.} & {$1/N$ } & {$2.485\pm0.207$}\tabularnewline
\hline 
{JL transform (Blocki et al. \cite{RefWorks:313}) } & {1.} & {$1/N$ } & {$2.394\pm0.216$}\tabularnewline
\hline 
{Laplace (Dwork et al. \cite{RefWorks:195}) } & {1.} & {0 } & {$2.634\pm0.174$}\tabularnewline
\hline 
{Exponential (Chaudhuri et al. \cite{RefWorks:178}) } & {1.} & {0 } & {$2.455\pm0.209$}\tabularnewline
\hline 
\end{tabular}
\par\end{centering}
\caption{Results from Exp. II: first principal component.} \label{tab:Exp2_results}
\end{table}

Table \ref{tab:Exp2_results} reports the results for Exp. II. Here
are the key observations. 
\begin{itemize}
\item Compared to the non-private baseline, the MVG mechanism
has reasonably small error $\Delta\rho$ of $0.2387$.
\item Compared to other $(\epsilon,\delta)$-basic mechanisms, i.e. the Gaussian mechanism
and the JL transform, the MVG mechanism provides better utility with the
same privacy guarantee (.01 and .0001 smaller error $\Delta\rho$,
respectively). 
\item Compared to other $\epsilon$-basic mechanisms, i.e. the Laplace and Exponential mechanisms,
the MVG mechanism yields higher utility with slightly weaker $(\epsilon,1/N)$-differential
privacy guarantee (.03 and .01 smaller error $\Delta\rho$, respectively). 
\end{itemize}
Overall, the MVG mechanism provides the best utility. Though we admit
that, with a weaker privacy guarantee, it does not provide significant
utility increase over the Exponential mechanism by Chaudhuri et al.
\cite{RefWorks:178}. Nevertheless, this method \cite{RefWorks:178} utilizes the positive-semidefinite (PSD) characteristic of the matrix query and is known to be among the best algorithms for this specific task. On the other hand, the MVG mechanism used in the experiment is more general. Furthermore, we show in the full version of this work that, when utilizing the PSD characteristic of the query function, the MVG mechanism can significantly outperform all three methods being compared here \cite{chanyaswad2018differential}. Again, in some applications, this trade-off of weaker privacy for better utility might be desirable \cite{RefWorks:418,RefWorks:419}, and
the MVG mechanism clearly provides the best trade-off.

\subsection{Experiment III: Covariance Estimation}

\begin{table}
\begin{centering}
\begin{tabular}{|c|c|c|c|}
\hline 
{Method } & {$\epsilon$ } & {$\delta$ } & {RSS ($\times10^{-2})$}\tabularnewline
\hline 
\hline 
{Non-private } & {- } & {- } & {$0$}\tabularnewline
\hline 
{Random guess } & {- } & {- } & {$\sim12.393$}\tabularnewline
\hline 
{MVG (Alg. \ref{alg:mvg_design_unimodal}) } & {1.} & {$1/N$ } & {$6.657\pm0.193$}\tabularnewline
\hline 
{Gaussian (Dwork et al. \cite{RefWorks:186}) } & {1.} & {$1/N$ } & {$7.029\pm0.216$}\tabularnewline
\hline 
{JL transform (Upadhyay \cite{RefWorks:399}) } & {1.} & {$1/N$ } & {$6.718\pm0.229$}\tabularnewline
\hline 
{Laplace (Dwork et al. \cite{RefWorks:195}) } & {1.} & {0 } & {$7.109\pm0.211$}\tabularnewline
\hline 
{Exponential (Blum et al. \cite{RefWorks:174}) } & {1.} & {0 } & {$7.223\pm0.211$}\tabularnewline
\hline 
\end{tabular}
\par\end{centering}
\caption{Results from Exp. III: covariance estimation.} \label{tab:Exp3_results}
\end{table}

\begin{figure*}
\begin{centering}
\includegraphics[scale=0.51]{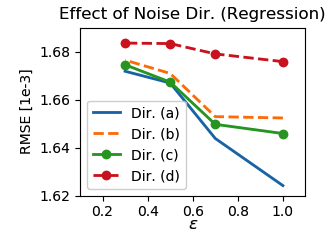}\includegraphics[scale=0.7]{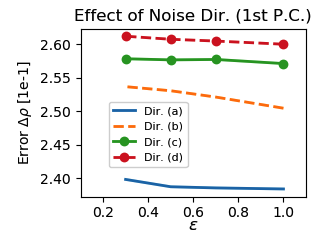}\includegraphics[scale=0.51]{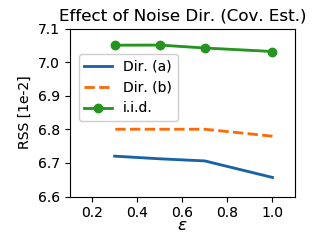}
\par\end{centering}
\caption{Effect of noise directions on the utility (all with $\delta=1/N$).
(Left) Exp. I: regression on the Liver dataset. The four directions shown put more precision
budget on the following features: (a) \{ALT, Y\}, (b) \{ALT\}, (c)
\{Y\}, (d) \{ALT, AST, Y\}. (Middle) Exp. II: $1^{st}$ P.C. on the Movement dataset. The four
directions emphasize the following features: (a) \{ANC0, ANC3\}, (b)
\{ANC0\}, (c) \{ANC3\}, (d) \{ANC1, ANC2\}. (Right) Exp. III: covariance
estimation on the CTG dataset. The two directions emphasize the two disjoint subsets
of features: (a) \{FHR, ASV, ALV\}, (b) The rest of the features.
\label{fig:Effect-of-noise-direction}}
\end{figure*}

Table \ref{tab:Exp3_results} reports the results for Exp. III. Here
are the key observations. 
\begin{itemize}
\item Compared to the non-private baseline, the MVG mechanism
has very small RSS error of $.06657$. 
\item Compared to other $(\epsilon,\delta)$-basic mechanisms, i.e. the Gaussian mechanism
and the JL transform, the MVG mechanism provides better utility with the
same privacy guarantee (.004 and .001 smaller RSS error, respectively). 
\item Compared to other $\epsilon$-basic mechanisms, i.e. the Laplace and Exponential mechanisms,
the MVG mechanism gives better utility with slightly weaker $(\epsilon,1/N)$-differential
privacy guarantee (.005 and .006 smaller RSS error, respectively). 
\end{itemize}
Overall, the MVG mechanism provides the best utility (smallest error).
When compared to other methods with stronger privacy guarantee, the
MVG mechanism can provide much higher utility. Again, we point out
that in some settings, the trade-off of weaker privacy for better
utility might be favorable \cite{RefWorks:418,RefWorks:419}, and
our approach provides the best trade-off.

\section{Discussion and Future Works}

\subsection{Effect of Directional Noise on Utility}

In Sec. \ref{subsec:Utility-Gain-via-dir-noise}, we discuss how the choice of noise directions can affect the utility. Here, we investigate this effect on the obtained utility in the three experiments. Fig. \ref{fig:Effect-of-noise-direction} depicts our results. 

Fig. \ref{fig:Effect-of-noise-direction}, Left, shows the direction comparison from Exp. I. We compare four choices of directions. Direction (a), which uses the domain knowledge (ALT) and the teacher label (Y), yields the best result when compared to: (b) using only the domain knowledge (ALT), (c) using only the teacher label (Y), and (d)
using an arbitrary extra feature (ALT+Y+AST).

Fig. \ref{fig:Effect-of-noise-direction}, Middle, shows the direction comparison from Exp. II. We compare four choices of directions. Direction (a), which makes full use of the prior information (ANC0 and ANC3),
performs best when compared to: (b), (c) using only partial prior
information (ANC0 or ANC3, respectively), and (d) having the wrong priors completely
(ANC1 and ANC2).

Fig. \ref{fig:Effect-of-noise-direction}, Right, shows the comparison from Exp. III. We compare three choices of directions. Direction (a), which uses the domain knowledge (FHR, ASV, ALV3), gives the best performance compared to: (b) using the completely wrong priors
(all other features), and (c) having no prior at all (i.i.d.).

The results illustrate three key points. First, as seen in all three plots in Fig. \ref{fig:Effect-of-noise-direction},
the choice of directions has an effect on the performance. Second,
as indicated by Fig. \ref{fig:Effect-of-noise-direction}, Right,
directional noise performs much better than i.i.d. noise. Third, as
seen in Fig. \ref{fig:Effect-of-noise-direction}, Left and Middle,
there may be multiple instances of directional noise that can lead
to comparable performance. The last observation shows the promise of the notion of directional noise, as it signals the robustness
of the approach.

\subsection{Runtime Comparison}

\begin{table}
\begin{centering}
\begin{tabular}{|>{\centering}p{4cm}|c|c|c|}
\hline 
\multirow{2}{*}{{\small{}Method}} & \multicolumn{3}{c|}{{\small{}Runtime (ms)}}\tabularnewline
\cline{2-4} 
 & {\small{}Exp. I} & {\small{}Exp. II} & {\small{}Exp. III}\tabularnewline
\hline 
\hline 
{\small{}MVG} & {\small{}36.2} & {\small{}1.8} & {\small{}$3.2\times10^{3}$}\tabularnewline
\hline 
{\small{}Gaussian (\cite{RefWorks:186}, \cite{RefWorks:249})} & {\small{}1.0} & {\small{}0.4} & {\small{}10.0}\tabularnewline
\hline 
{\small{}JL transform (\cite{RefWorks:399}, \cite{RefWorks:313})} & {\small{}192.7} & {\small{}637.4} & {\small{}$2.57\times10^{6}$}\tabularnewline
\hline 
{\small{}Laplace (\cite{RefWorks:195})} & {\small{}0.4} & {\small{}0.5} & {\small{}8.0}\tabularnewline
\hline 
{\small{}Exponential (\cite{RefWorks:174}, \cite{RefWorks:178})} & {\small{}627.2} & {\small{}2,112.7} & {\small{}$2.00\times10^{6}$}\tabularnewline
\hline 
\end{tabular}
\par\end{centering}
\caption{Runtime of each method on the three experiments.\label{tab:Runtime}}

\end{table}

Next, we present the empirical runtime comparison between the MVG mechanism and the compared methods in Table \ref{tab:Runtime}. The experiments are run on an AMD Opteron 6320 Processor with 4 cores using Python 2.7, along with NumPy \cite{oliphant2006guide}, SciPy \cite{scipy}, Scikit-learn \cite{RefWorks:231}, and emcee \cite{foreman2013emcee} packages. The results show that, although the MVG mechanism runs much slower than the Gaussian and Laplace mechanisms, it runs much faster than the JL transform and the Exponential mechanism.

Both observations are as expected. First, the MVG mechanism runs slower than the i.i.d.-based Gaussian and Laplace mechanisms because it incurs the computational overhead of deriving the non-i.i.d. noise. The amount of overhead depends on the size of the query output as discussed in Sec. \ref{subsec:Sampling-from-mvg}. Second, the MVG mechanism runs much faster than the JL transform because, in addition to requiring SVD to modify the singular values of the matrix query and i.i.d. Gaussian samples similar to the MVG mechanism, the JL transform has a runtime overhead for the construction of its projection matrix, which consists of multiple matrix multiplications. Finally, the MVG mechanism runs much faster than the Exponential mechanism since drawing samples from the distribution defined by the Exponential mechanism may not be efficient.

\subsection{Directional Noise as a Generalized Subspace Projection}

Directional noise provides utility gain by adding less noise in useful
directions and more in the others. This has a connection to subspace projection or dimensionality
reduction, which simply removes the non-useful directions. The main
difference between the two is that, in directional noise, the non-useful
directions are kept, although are highly perturbed. However, despite
being highly perturbed, these directions may still be able to contribute
to the utility performance. 

We test this hypothesis by running two
additional regression experiments (Exp. I) as follows. Given the same
two important features (ALT \& Y), we use the Gaussian mechanism \cite{RefWorks:186}
and the JL transform method \cite{RefWorks:399} to perform the regression
task using only these two features. With $\epsilon=1$ and $\delta=1/N$,
the results are $(2.538\pm.065)\times10^{-2}$ and $(2.863\pm.022)\times10^{-2}$
of RMSE, respectively. Noticeably, these results are significantly
worse than that of the MVG mechanism, with the same privacy guarantee.
Specifically, by incorporating all features with directional noise
via the MVG mechanism, we can achieve over 150\% gain in utility over
the dimensionality reduction alternatives.

\subsection{Exploiting Structural Characteristics of the Matrices}

In this work, we derive the sufficient condition for the MVG mechanism
without making any assumption on the query function. However, many
practical matrix-valued query functions have a specific structure,
e.g. the covariance matrix is positive semi definite (PSD) \cite{RefWorks:208},
the Laplacian matrix \cite{RefWorks:329} is symmetric.
Therefore, future works may look into exploiting these intrinsic characteristics of the matrices
in the derivation of the differential-privacy condition for the MVG mechanism.

 \subsection{Precision Allocation Strategy Design}

 \label{subsec:utility_analysis_via_mi}

Alg. \ref{alg:mvg_design_unimodal} and Alg. \ref{alg:mvg_design_equimodal}
 take as an input the precision allocation strategy vector $\boldsymbol{\theta}\in(0,1)^{m}$:$\bigl|\boldsymbol{\theta}\bigr|_{1}=1$.
 Elements of $\boldsymbol{\theta}$ are chosen to emphasize how informative
 or useful each direction is. The design of $\boldsymbol{\theta}$
 to optimize the utility gain via the directional noise is an interesting
 topic for future research. For example, in our experiments, we use
 the intuition that our prior knowledge only tells us whether the directions
 are highly informative or not, but we \emph{do not know} the granularity
 of the level of usefulness of these directions. Hence, we adopt the
 \emph{binary allocation} strategy, i.e. give most precision budget
 to the useful directions in equal amount, and give the rest of the
 budget to the other directions in equal amount. An interesting direction for future work is to 
 investigate general instances
 when the knowledge about the directions is more granular.

\section{Conclusion}

We study the matrix-valued query function in differential privacy,
and present the MVG mechanism that is designed specifically for this type of query function.
We prove that the MVG mechanism guarantees
($\epsilon,\delta$)-differential privacy, and, consequently, introduce the novel concept
of directional noise, which can be used to reduce the impact of the
noise on the utility of the query answer. Finally, we evaluate our
approach experimentally for three matrix-valued query functions on
three privacy-sensitive datasets, and the results show that our approach
can provide the utility improvement over existing methods in all of the experiments.

\section*{Acknowledgement}

The authors would like to thank the reviewers for their valuable feedback that helped improve the paper. This work was supported in part by the National Science Foundation (NSF) under Grants CNS-1553437, CCF-1617286, and CNS-1702808; an Army Research Office YIP Award; and faculty research awards from Google, Cisco, Intel, and IBM.

\bibliographystyle{ACM-Reference-Format}
\bibliography{mvg_references}

\appendix

\section{Full Proof of \texorpdfstring{$(\epsilon,\delta)$}{(epsilon,delta)}-Differential Privacy} \label{sec:Full-Proof}

We present the full proof of the sufficient condition for the MVG
mechanism to guarantee $(\epsilon,\delta)$-differential privacy presented in Theorem
\ref{thm:design_general} here. 
\begin{proof}
The MVG mechanism guarantees differential privacy if for every pair
of neighboring datasets $\{\mathbf{X}_{1},\mathbf{X}_{2}\}$ and all
possible measurable sets $\mathbf{S}\subseteq\mathbb{R}^{m\times n}$,
\[
\Pr\left[f(\mathbf{X}_{1})+\mathcal{Z}\in\mathbf{S}\right]\leq\exp(\epsilon)\Pr\left[f(\mathbf{X}_{2})+\mathcal{Z}\in\mathbf{S}\right].
\]
The proof now follows by observing that (Lemma \ref{lem:affine_tx}),
\[
\mathcal{Z}=\mathbf{W}_{\boldsymbol{\Sigma}}\boldsymbol{\Lambda}_{\boldsymbol{\Sigma}}^{1/2}\mathcal{N}\boldsymbol{\Lambda}_{\boldsymbol{\Psi}}^{1/2}\mathbf{W}_{\boldsymbol{\Psi}}^{T},
\]
and defining the following events: 
\[
\mathbf{R}_{1}=\{\mathcal{N}:\|\mathcal{N}\|_{F}^{2}\le\zeta(\delta)^{2}\},\,\mathbf{R}_{2}=\{\mathcal{N}:\|\mathcal{N}\|_{F}^{2}>\zeta(\delta)^{2}\},
\]
where $\zeta(\delta)$ is defined in Theorem \ref{thm:laurent_massart}.
Next, observe that 
\begin{align*}
 & \Pr\left[f(\mathbf{X}_{1})+\mathcal{Z}\in\mathbf{S}\right]\\
 & =\Pr\left[\left(\{f(\mathbf{X}_{1})+\mathcal{Z}\in\mathbf{S}\}\cap\mathbf{R}_{1}\right)\cup\left(\{f(\mathbf{X}_{1})+\mathcal{Z}\in\mathbf{S}\}\cap\mathbf{R}_{2}\right)\right]\\
 & \le\Pr\left[\{f(\mathbf{X}_{1})+\mathcal{Z}\in\mathbf{S}\}\cap\mathbf{R}_{1}\right]+\Pr\left[\{f(\mathbf{X}_{1})+\mathcal{Z}\in\mathbf{S}\}\cap\mathbf{R}_{2}\right],
\end{align*}
where the last inequality follows from the union bound. By Theorem
\ref{thm:laurent_massart} and the definition of the set $\mathbf{R}_{2}$,
we have, 
\begin{align*}
\Pr\left[\{f(\mathbf{X}_{1})+\mathcal{Z}\in\mathbf{S}\}\cap\mathbf{R}_{2}\right]\le\Pr\left[\mathbf{R}_{2}\right]=1-\Pr\left[\mathbf{R}_{1}\right]\le\delta.
\end{align*}

In the rest of the proof, we find sufficient conditions for the following
inequality to hold: 
\begin{align*}
\Pr\left[f(\mathbf{X}_{1})+\mathcal{Z}\in(\mathbf{S}\cap\mathbf{R}_{1})\right]\le\exp(\epsilon)\Pr\left[f(\mathbf{X}_{2})+\mathcal{Z}\in\mathbf{S}\right].
\end{align*}
this would complete the proof of differential privacy guarantee.

Using the definition of $\mathcal{MVG}_{m,n}(\mathbf{0},\boldsymbol{\Sigma},\boldsymbol{\Psi})$
(Definition \ref{def:tmvg_dist}), this is satisfied if we have, 
\begin{align*}
\int_{\mathbf{S}\cap\mathbf{R}_{1}}e^{\{-\frac{1}{2}\mathrm{tr}[\boldsymbol{\Psi}^{-1}(\mathbf{Y}-f(\mathbf{X}_{1}))^{T}\boldsymbol{\Sigma}^{-1}(\mathbf{Y}-f(\mathbf{X}_{1}))]\}}d\mathbf{Y} & \leq\\
e^{\epsilon}\int_{\mathbf{S}\cap\mathbf{R}_{1}}e^{\{-\frac{1}{2}\mathrm{tr}[\boldsymbol{\Psi}^{-1}(\mathbf{Y}-f(\mathbf{X}_{2}))^{T}\boldsymbol{\Sigma}^{-1}(\mathbf{Y}-f(\mathbf{X}_{2}))]\}}d\mathbf{Y}.
\end{align*}
By inserting $\frac{\exp\{-\frac{1}{2}\mathrm{tr}[\boldsymbol{\Psi}^{-1}(\mathbf{Y}-f(\mathbf{X}_{2}))^{T}\boldsymbol{\Sigma}^{-1}(\mathbf{Y}-f(\mathbf{X}_{2}))]\}}{\exp\{-\frac{1}{2}\mathrm{tr}[\boldsymbol{\Psi}^{-1}(\mathbf{Y}-f(\mathbf{X}_{2}))^{T}\boldsymbol{\Sigma}^{-1}(\mathbf{Y}-f(\mathbf{X}_{2}))]\}}$
inside the integral on the left side, it is sufficient to show that
\[
\frac{\exp\{-\frac{1}{2}\mathrm{tr}[\boldsymbol{\Psi}^{-1}(\mathbf{Y}-f(\mathbf{X}_{1}))^{T}\boldsymbol{\Sigma}^{-1}(\mathbf{Y}-f(\mathbf{X}_{1}))]\}}{\exp\{-\frac{1}{2}\mathrm{tr}[\boldsymbol{\Psi}^{-1}(\mathbf{Y}-f(\mathbf{X}_{2}))^{T}\boldsymbol{\Sigma}^{-1}(\mathbf{Y}-f(\mathbf{X}_{2}))]\}}\leq\exp(\epsilon),
\]
for all $\mathbf{Y}\in\mathbf{S}\cap\mathbf{R}_{1}$. With some algebraic
manipulations, the left hand side of this condition can be expressed
as, 
\begin{align*}
= & \exp\{-\frac{1}{2}\mathrm{tr}[\boldsymbol{\Psi}^{-1}\mathbf{Y}{}^{T}\mathbf{\boldsymbol{\Sigma}}^{-1}(f(\mathbf{X}_{2})-f(\mathbf{X}_{1}))\\
 & +\boldsymbol{\Psi}^{-1}(f(\mathbf{X}_{2})-f(\mathbf{X}_{1}))^{T}\boldsymbol{\Sigma}^{-1}\mathbf{Y}]\}\\
 & -\boldsymbol{\Psi}^{-1}f(\mathbf{X}_{2})^{T}\boldsymbol{\Sigma}^{-1}f(\mathbf{X}_{2})+\boldsymbol{\Psi}^{-1}f(\mathbf{X}_{1})^{T}\boldsymbol{\Sigma}^{-1}f(\mathbf{X}_{1})\\
= & \exp\{\frac{1}{2}\mathrm{tr}[\boldsymbol{\Psi}^{-1}\mathbf{Y}^{T}\mathbf{\boldsymbol{\Sigma}}^{-1}\boldsymbol{\Delta}+\boldsymbol{\Psi}^{-1}\boldsymbol{\Delta}^{T}\boldsymbol{\Sigma}^{-1}\mathbf{Y}\\
 & +\boldsymbol{\Psi}^{-1}f(\mathbf{X}_{2})^{T}\boldsymbol{\Sigma}^{-1}f(\mathbf{X}_{2})-\boldsymbol{\Psi}^{-1}f(\mathbf{X}_{1})^{T}\boldsymbol{\Sigma}^{-1}f(\mathbf{X}_{1})]\},
\end{align*}
where $\boldsymbol{\Delta}=f(\mathbf{X}_{1})-f(\mathbf{X}_{2})$.
This quantity has to be bounded by $\leq\exp(\epsilon)$, so we present
the following \emph{characteristic equation}, which has to be satisfied
for all possible neighboring $\{\mathbf{X}_{1},\mathbf{X}_{2}\}$
and all $\mathbf{Y}\in\mathbf{S}\cap\mathbf{R}_{1}$, for the MVG
mechanism to guarantee $(\epsilon,\delta)$-differential privacy:
\begin{align*}
\mathrm{tr}[\boldsymbol{\Psi}^{-1}\mathcal{Y}^{T}\boldsymbol{\Sigma}^{-1}\boldsymbol{\Delta}+\boldsymbol{\Psi}^{-1}\boldsymbol{\Delta}^{T}\boldsymbol{\Sigma}^{-1}\mathcal{Y}\\
+\boldsymbol{\Psi}^{-1}f(\mathbf{X}_{2})^{T}\boldsymbol{\Sigma}^{-1}f(\mathbf{X}_{2})-\boldsymbol{\Psi}^{-1}f(\mathbf{X}_{1})^{T}\boldsymbol{\Sigma}^{-1}f(\mathbf{X}_{1})] & \leq2\epsilon.
\end{align*}
Specifically, we want to show that this inequality holds with probability
$1-\delta$.

From the characteristic equation, the proof analyzes the four terms
in the sum separately since the trace is additive.

\emph{The first term}: $\mathrm{tr}[\boldsymbol{\Psi}^{-1}\mathbf{Y}^{T}\boldsymbol{\Sigma}^{-1}\boldsymbol{\Delta}]$.
First, let us denote $\mathcal{Y}=f(\mathbf{X})+\mathcal{Z}$, where
$f(\mathbf{X})$ and $\mathcal{Z}$ are any possible instances of
the query and the noise, respectively. Then, we can rewrite the first
term as, $\mathrm{tr}[\boldsymbol{\Psi}^{-1}f(\mathbf{X})^{T}\boldsymbol{\Sigma}^{-1}\boldsymbol{\Delta}]+\mathrm{tr}[\boldsymbol{\Psi}^{-1}\mathcal{Z}^{T}\boldsymbol{\Sigma}^{-1}\boldsymbol{\Delta}]$.
The earlier part can be bounded from Lemma \ref{lem:v_neumann}: 
\[
\mathrm{tr}[\boldsymbol{\Psi}^{-1}f(\mathbf{X})^{T}\boldsymbol{\Sigma}^{-1}\boldsymbol{\Delta}]\leq\sum_{i=1}^{r}\sigma_{i}(\boldsymbol{\Psi}^{-1}f(\mathbf{X})^{T})\sigma_{i}(\boldsymbol{\Delta}^{T}\boldsymbol{\Sigma}^{-1}).
\]
Lemma \ref{lem:singular_bound} can then be used to bound each singular
value. In more detail, 
\[
\sigma_{i}(\boldsymbol{\Psi}^{-1}f(\mathbf{X})^{T})\leq\frac{\left\Vert \boldsymbol{\Psi}^{-1}f(\mathbf{X})^{T}\right\Vert _{F}}{\sqrt{i}}\leq\frac{\left\Vert \boldsymbol{\Psi}^{-1}\right\Vert _{F}\left\Vert f(\mathbf{X})\right\Vert _{F}}{\sqrt{i}},
\]
where the last inequality is via the sub-multiplicative property of
a matrix norm \cite{RefWorks:290}. It is well-known that $\bigl\Vert\boldsymbol{\Psi}^{-1}\bigr\Vert_{F}=\bigl\Vert\boldsymbol{\sigma}(\boldsymbol{\Psi}^{-1})\bigr\Vert_{2}$
(cf. \cite[p. 342]{RefWorks:208}), and since $\gamma=\sup_{\mathbf{X}}\bigl\Vert f(\mathbf{X})\bigr\Vert_{F}$,
\[
\sigma_{i}(\boldsymbol{\Psi}^{-1}f(\mathbf{X})^{T})\leq\gamma\left\Vert \boldsymbol{\sigma}(\boldsymbol{\Psi}^{-1})\right\Vert _{2}/\sqrt{i}.
\]
Applying the same steps to the other singular value, and using Definition
\ref{def:sensitivity}, we can write, 
\[
\sigma_{i}(\boldsymbol{\Delta}^{T}\boldsymbol{\Sigma}^{-1})\leq s_{2}(f)\left\Vert \boldsymbol{\sigma}(\boldsymbol{\Sigma}^{-1})\right\Vert _{2}/\sqrt{i}.
\]
By substituting the two singular value bounds, the earlier part of the
first term can be bounded by, 
\begin{align}
 & \mathrm{tr}[\boldsymbol{\Psi}^{-1}f(\mathbf{X})^{T}\boldsymbol{\Sigma}^{-1}\boldsymbol{\Delta}]\nonumber \\
\leq & \gamma s_{2}(f)H_{r}\left\Vert \boldsymbol{\sigma}(\boldsymbol{\Sigma}^{-1})\right\Vert _{2}\left\Vert \boldsymbol{\sigma}(\boldsymbol{\Psi}^{-1})\right\Vert _{2}.\label{eq:first_term_part1}
\end{align}

The latter part of the first term is more complicated since it involves
$\mathcal{Z}$, so we will derive the bound in more detail. First,
let us define $\mathcal{N}$ to be drawn from $\mathcal{MVG}_{m,n}(\mathbf{0},\mathbf{I}_{m},\mathbf{I}_{n})$,
so we can write $\mathcal{Z}$ in terms of $\mathcal{N}$ using affine
transformation \cite{RefWorks:279}: $\mathcal{Z}=\mathbf{B}_{\boldsymbol{\Sigma}}\mathcal{N}\mathbf{B}_{\boldsymbol{\Psi}}^{T}$.
To specify $\mathbf{B}_{\boldsymbol{\Sigma}}$ and $\mathbf{B}_{\boldsymbol{\Psi}}$,
we solve the following linear equations, respectively, 
\[
\mathbf{B}_{\boldsymbol{\Sigma}}\mathbf{B}_{\boldsymbol{\Sigma}}^{T}=\boldsymbol{\Sigma};\mathbf{B}_{\boldsymbol{\Psi}}\mathbf{B}_{\boldsymbol{\Psi}}^{T}=\boldsymbol{\Psi}.
\]
This can be readily solved with SVD (cf. \cite[p. 440]{RefWorks:208});
hence, $\mathbf{B}_{\boldsymbol{\Sigma}}=\mathbf{W}_{\boldsymbol{\Sigma}}\boldsymbol{\Lambda}_{\boldsymbol{\Sigma}}^{\frac{1}{2}}$,
and $\mathbf{B}_{\boldsymbol{\Psi}}=\mathbf{W}_{\boldsymbol{\Psi}}\boldsymbol{\Lambda}_{\boldsymbol{\Psi}}^{\frac{1}{2}}$,
where $\boldsymbol{\Sigma}=\mathbf{W}_{\boldsymbol{\Sigma}}\boldsymbol{\Lambda}_{\boldsymbol{\Sigma}}\mathbf{W}_{\boldsymbol{\Sigma}}^{T}$,
and $\boldsymbol{\Psi}=\mathbf{W}_{\boldsymbol{\Psi}}\boldsymbol{\Lambda}_{\boldsymbol{\Psi}}\mathbf{W}_{\boldsymbol{\Psi}}^{T}$
from SVD. Therefore, $\mathcal{Z}$ can be written as, 
\[
\mathcal{Z}=\mathbf{W}_{\boldsymbol{\Sigma}}\boldsymbol{\Lambda}_{\boldsymbol{\Sigma}}^{1/2}\mathcal{N}\boldsymbol{\Lambda}_{\boldsymbol{\Psi}}^{1/2}\mathbf{W}_{\boldsymbol{\Psi}}^{T}.
\]
Substituting into the latter part of the first term yields, 
\[
\mathrm{tr}[\boldsymbol{\Psi}^{-1}\mathcal{Z}^{T}\boldsymbol{\Sigma}^{-1}\boldsymbol{\Delta}]=\mathrm{tr}[\mathbf{W}_{\boldsymbol{\Psi}}\boldsymbol{\Lambda}_{\boldsymbol{\Psi}}^{-1/2}\mathcal{N}\boldsymbol{\Lambda}_{\boldsymbol{\Sigma}}^{-1/2}\mathbf{W}_{\boldsymbol{\Sigma}}^{T}\boldsymbol{\Delta}].
\]
This can be bounded by Lemma \ref{lem:v_neumann} as, 
\[
\mathrm{tr}[\boldsymbol{\Psi}^{-1}\mathcal{Z}^{T}\boldsymbol{\Sigma}^{-1}\boldsymbol{\Delta}]\leq\sum_{i=1}^{r}\sigma_{i}(\mathbf{W}_{\boldsymbol{\Psi}}\boldsymbol{\Lambda}_{\boldsymbol{\Psi}}^{-1/2}\mathcal{N}\boldsymbol{\Lambda}_{\boldsymbol{\Sigma}}^{-1/2}\mathbf{W}_{\boldsymbol{\Sigma}}^{T})\sigma_{i}(\boldsymbol{\Delta}).
\]
The two singular values can then be bounded by Lemma \ref{lem:singular_bound}.
For the first singular value, 
\begin{align*}
\sigma_{i}(\mathbf{W}_{\boldsymbol{\Psi}}\boldsymbol{\Lambda}_{\boldsymbol{\Psi}}^{-1/2}\mathcal{N}\boldsymbol{\Lambda}_{\boldsymbol{\Sigma}}^{-1/2}\mathbf{W}_{\boldsymbol{\Sigma}}^{T}) & \leq\frac{\left\Vert \mathbf{W}_{\boldsymbol{\Psi}}\boldsymbol{\Lambda}_{\boldsymbol{\Psi}}^{-1/2}\mathcal{N}\boldsymbol{\Lambda}_{\boldsymbol{\Sigma}}^{-1/2}\mathbf{W}_{\boldsymbol{\Sigma}}^{T}\right\Vert _{F}}{\sqrt{i}}\\
 & \leq\frac{\left\Vert \boldsymbol{\Lambda}_{\boldsymbol{\Sigma}}^{-1/2}\right\Vert _{F}\left\Vert \boldsymbol{\Lambda}_{\boldsymbol{\Psi}}^{-1/2}\right\Vert _{F}\left\Vert \mathcal{N}\right\Vert _{F}}{\sqrt{i}}.
\end{align*}
By definition, $\left\Vert \boldsymbol{\Lambda}_{\boldsymbol{\Sigma}}^{-1/2}\right\Vert _{F}=\sqrt{\mathrm{tr}(\boldsymbol{\Lambda}_{\boldsymbol{\Sigma}}^{-1})}=\left\Vert \boldsymbol{\sigma}(\boldsymbol{\Sigma}^{-1})\right\Vert _{1}^{1/2}$,
where $\left\Vert \cdot\right\Vert _{1}$ is the 1-norm. By norm relation,
$\left\Vert \boldsymbol{\sigma}(\boldsymbol{\Sigma}^{-1})\right\Vert _{1}^{1/2}\leq m^{1/4}\left\Vert \boldsymbol{\sigma}(\boldsymbol{\Sigma}^{-1})\right\Vert _{2}^{1/2}$.
With similar derivation for $\bigl\Vert\boldsymbol{\Lambda}_{\boldsymbol{\Psi}}^{-1/2}\bigr\Vert_{F}$
and with Theorem \ref{thm:laurent_massart}, the singular value can
be bounded with probability $1-\delta$ as, 
\begin{align*}
 & \sigma_{i}(\mathbf{W}_{\boldsymbol{\Psi}}\boldsymbol{\Lambda}_{\boldsymbol{\Psi}}^{-\frac{1}{2}}\mathcal{N}\boldsymbol{\Lambda}_{\boldsymbol{\Sigma}}^{-\frac{1}{2}}\mathbf{W}_{\boldsymbol{\Sigma}}^{T})\\
\leq & \frac{(mn)^{\frac{1}{4}}\zeta(\delta)\left\Vert \boldsymbol{\sigma}(\boldsymbol{\Sigma}^{-1})\right\Vert _{2}^{\frac{1}{2}}\left\Vert \boldsymbol{\sigma}(\boldsymbol{\Psi}^{-1})\right\Vert _{2}^{\frac{1}{2}}}{\sqrt{i}}.
\end{align*}
Meanwhile, the other singular value can be readily bounded with Lemma
\ref{lem:singular_bound} as $\sigma_{i}(\boldsymbol{\Delta})\leq s_{2}(f)/\sqrt{i}$.
Hence, the latter part of the first term is bounded with probability
$\geq1-\delta$ as, 
\begin{align}
 & \mathrm{tr}[\boldsymbol{\Psi}^{-1}\mathbf{Z}^{T}\boldsymbol{\Sigma}^{-1}\boldsymbol{\Delta}]\nonumber \\
\leq & (mn)^{\frac{1}{4}}\zeta(\delta)H_{r}s_{2}(f)\left\Vert \boldsymbol{\sigma}(\boldsymbol{\Sigma}^{-1})\right\Vert _{2}^{\frac{1}{2}}\left\Vert \boldsymbol{\sigma}(\boldsymbol{\Psi}^{-1})\right\Vert _{2}^{\frac{1}{2}}.\label{eq:first_term_part2}
\end{align}
Since the parameter $(\left\Vert \boldsymbol{\sigma}(\boldsymbol{\Sigma}^{-1})\right\Vert _{2}\left\Vert \boldsymbol{\sigma}(\boldsymbol{\Psi}^{-1})\right\Vert _{2})^{1/2}$
appears a lot in the derivation, let us define 
\[
\phi=(\left\Vert \boldsymbol{\sigma}(\boldsymbol{\Sigma}^{-1})\right\Vert _{2}\left\Vert \boldsymbol{\sigma}(\boldsymbol{\Psi}^{-1})\right\Vert _{2})^{1/2}.
\]
Finally, combining Eq. (\ref{eq:first_term_part1}) and (\ref{eq:first_term_part2})
yields the bound for the first term, 
\[
\mathrm{tr}[\boldsymbol{\Psi}^{-1}\mathbf{Y}^{T}\boldsymbol{\Sigma}^{-1}\boldsymbol{\Delta}]\leq\gamma H_{r}s_{2}(f)\phi^{2}+(mn)^{1/4}\zeta(\delta)H_{r}s_{2}(f)\phi.
\]

\emph{The second term}: $\mathrm{tr}[\boldsymbol{\Psi}^{-1}\boldsymbol{\Delta}^{T}\boldsymbol{\Sigma}^{-1}\mathbf{Y}]$.
By following the same steps as in the first term, it can be shown
that the second term has the exact same bound as the first terms,
i.e. 
\[
\mathrm{tr}[\boldsymbol{\Psi}^{-1}\boldsymbol{\Delta}^{T}\boldsymbol{\Sigma}^{-1}\mathbf{Y}]\leq\gamma H_{r}s_{2}(f)\phi^{2}+(mn)^{1/4}\zeta(\delta)H_{r}s_{2}(f)\phi.
\]

\emph{The third term}: $\mathrm{tr}[\boldsymbol{\Psi}^{-1}f(\mathbf{X}_{2})^{T}\boldsymbol{\Sigma}^{-1}f(\mathbf{X}_{2})]$.
Applying Lemma \ref{lem:v_neumann} and \ref{lem:singular_bound},
we can readily bound it as, 
\[
\mathrm{tr}[\boldsymbol{\Psi}^{-1}f(\mathbf{X}_{2})^{T}\boldsymbol{\Sigma}^{-1}f(\mathbf{X}_{2})]\leq\gamma^{2}H_{r}\phi^{2}.
\]

\emph{The fourth term}: $-\mathrm{tr}[\boldsymbol{\Psi}^{-1}f(\mathbf{X}_{1})^{T}\boldsymbol{\Sigma}^{-1}f(\mathbf{X}_{1})]$.
Since this term has the negative sign, we consider the absolute value
instead. Using Lemma \ref{lem:abs_trace_bound}, 
\[
\left|\mathrm{tr}[\boldsymbol{\Psi}^{-1}f(\mathbf{X}_{1})^{T}\boldsymbol{\Sigma}^{-1}f(\mathbf{X}_{1})]\right|\leq\sum_{i=1}^{r}\sigma_{i}(\boldsymbol{\Psi}^{-1}f(\mathbf{X}_{1})^{T}\boldsymbol{\Sigma}^{-1}f(\mathbf{X}_{1})).
\]
Then, using the singular value bound in Lemma \ref{lem:singular_bound},
\[
\sigma_{i}(\boldsymbol{\Psi}^{-1}f(\mathbf{X}_{1})^{T}\boldsymbol{\Sigma}^{-1}f(\mathbf{X}_{1}))\leq\frac{\left\Vert \boldsymbol{\Psi}^{-1}\right\Vert _{F}\left\Vert f(\mathbf{X}_{1})\right\Vert _{F}^{2}\left\Vert \boldsymbol{\Sigma}^{-1}\right\Vert _{F}}{\sqrt{i}}.
\]
Hence, the fourth term can be bounded by, 
\[
\left|\mathrm{tr}[\boldsymbol{\Psi}^{-1}f(\mathbf{X}_{1})^{T}\boldsymbol{\Sigma}^{-1}f(\mathbf{X}_{1})]\right|\leq\gamma^{2}H_{r,1/2}\phi^{2}.
\]

\emph{Four terms combined:} by combining the four terms and rearranging
them, the characteristic equation becomes, 
\[
\alpha\phi^{2}+\beta\phi\leq2\epsilon.
\]
This is a quadratic equation, of which the solution is known to be $\phi\in[\frac{-\beta-\sqrt{\beta^{2}+8\alpha\epsilon}}{2\alpha},\frac{-\beta+\sqrt{\beta^{2}+8\alpha\epsilon}}{2\alpha}]$.
Since we know $\phi>0$, we only have the one-sided solution, 
\[
\phi\leq\frac{-\beta+\sqrt{\beta^{2}+8\alpha\epsilon}}{2\alpha},
\]
which implies the criterion in Theorem \ref{thm:design_general}. 
\end{proof}

\section{\texorpdfstring{$L_{2}$}{L2}-Sensitivities}\label{sec:L2-Sensitivities}

Here, we derive the $L_{2}$-sensitivity for the MVG mechanism used
in our experiments.

\subsection{Experiment I}

The query function is $f(\mathbf{X})=\mathbf{X}\in[0,1]^{6\times248}$.
For neighboring datasets $\{\mathbf{X},\mathbf{X}'\}$, the $L_{2}$-sensitivity
is 
\[
s_{2}(f)=\sup_{\mathbf{X},\mathbf{X}'}\left\Vert \mathbf{X}-\mathbf{X}'\right\Vert _{F}=\sup_{\mathbf{X},\mathbf{X}'}\sqrt{\sum_{i=1}^{6}(x(i)-x'(i))^{2}}=\sqrt{6}.
\]

\subsection{Experiment II}

The query function is $f(\mathbf{X})=\frac{1}{N}\mathbf{X}\mathbf{X}^{T}$,
where $\mathbf{X}\in[-1,1]^{4\times2021}$. For neighboring
datasets $\{\mathbf{X},\mathbf{X}'\}$, the $L_{2}$-sensitivity is
\[
s_{2}(f)=\sup_{\mathbf{X},\mathbf{X}'}\frac{\left\Vert \mathbf{x}_{j}\mathbf{x}_{j}^{T}-\mathbf{x}_{j}'\mathbf{x}_{j}'^{T}\right\Vert _{F}}{2021}=\frac{2\sqrt{\sum_{j=1}^{4^{2}}x_{j}(i)^{2}}}{2021}=\frac{8}{2021}.
\]

\subsection{Experiment III}

The query function is the same as Exp. I, so the $L_{2}$-sensitivity
can be readily derived as $s_{2}(f)=\sup_{\mathbf{X},\mathbf{X}'}\sqrt{\sum_{i=1}^{21}(x(i)-x'(i))^{2}}=\sqrt{21}$.

\end{document}